\documentclass[a4paper,twocolumn,11pt,received=2024-10-19]{quantumarticle}
\pdfoutput=1
\usepackage[utf8]{inputenc}
\usepackage[english]{babel}
\usepackage[T1]{fontenc}
\usepackage{amsmath}

\usepackage{graphicx}
\usepackage{hyperref}
\usepackage{amssymb}
\usepackage{tikz}
\usepackage{lipsum}

\usepackage{epstopdf}
\usepackage{mathtools}
\usepackage{qcircuit}
\usepackage{braket}
\usepackage{amsthm}
\usepackage{multirow}
\usepackage{marvosym}
\usepackage{cite}
\usepackage{makecell}
\usepackage{pifont}
\usepackage{threeparttable}
\usepackage[caption=false,font=scriptsize,labelfont=sf,textfont=sf]{subfig}
\newtheorem{theorem}{Theorem}
\newtheorem{definition}{Definition}
\newtheorem{lemma}{Lemma}

\begin{document}

\title{Quantum Indistinguishable Obfuscation via Quantum Circuit Equivalence}

\author{Yuanjing Zhang}
\affiliation{School of Cyber Science and Technology, Beihang University, Beijing 100191, China}
\orcid{0009-0007-2739-6075}
\author{Tao Shang\textsuperscript{\Letter}}
\email{shangtao@buaa.edu.cn}
\orcid{0000-0003-2369-1521}
\thanks{This project was supported by the National Natural Science Foundation of China (No. 62471020, 61971021), the Key Research and Development Program of Hebei Province (No. 22340701D), and the Chinese Universities Industry-Education-Research Innovation Foundation of BII Education Grant Program (No. 2021BCA0200) for valuable helps. }
\affiliation{School of Cyber Science and Technology, Beihang University, Beijing 100191, China}
\author{Kun Zhang}
\affiliation{School of Cyber Science and Technology, Beihang University, Beijing 100191, China}
\author{Chenyi Zhang}
\affiliation{School of Cyber Science and Technology, Beihang University, Beijing 100191, China}
\author{Haohua Du}
\affiliation{School of Cyber Science and Technology, Beihang University, Beijing 100191, China}
\author{Xueyi Guo}
\affiliation{Beijing Academy of Quantum Information Sciences Beijing, China}
\maketitle

\begin{abstract}
Quantum computing solutions are increasingly deployed in commercial environments through delegated computing, especially one of the most critical issues is to guarantee the confidentiality and proprietary  of quantum implementations. Since the proposal of general-purpose indistinguishability obfuscation (iO) and functional encryption schemes, iO has emerged as a seemingly versatile cryptography primitive. Existing research on quantum indistinguishable obfuscation (QiO) primarily focuses on task-oriented, lacking solutions to general quantum computing. In this paper, we propose a scheme for constructing QiO via the equivalence of quantum circuits. It introduces the concept of quantum subpath sum equivalence, demonstrating that indistinguishability between two quantum circuits can be achieved by incremental changes in quantum subpaths. The restriction of security loss is solved by reducing the distinguisher to polynomial probability test. The scheme obfuscates the quantum implementation of classical functions in a path-sum specification, ensuring the indistinguishability between different quantum implementations. 
The results demonstrate the feasibility of indistinguishability obfuscation for general circuits and provide novel insights on intellectual property protection and secure delegated quantum computing. 
\end{abstract}

\section{Introduction}
In recent years, quantum computing has emerged as a transformative technology with the potential to revolutionize various fields, including scientific research, optimization, and machine learning. As quantum computing solutions are increasingly deployed in commercial settings, one of the foremost concerns is to guarantee the confidentiality and proprietary of quantum circuits, particularly in the context of delegated quantum computing. The circuits encapsulate complex algorithms, sensitive data, and proprietary functionalities essential for executing various computational tasks. It is crucial to ensure the integrity and confidentiality of quantum circuits for safeguarding intellectual property, trade secrets, and sensitive information. 

The concept of obfuscation originated from program obfuscation: software companies aim to prevent users from obtaining the implementation logic of the program through decompilation and other methods after releasing their software, thereby safeguarding their trade secrets. The theoretical foundation of program obfuscation was proposed in 2001 \cite{2001Barak}, and it was not until 2013 that researchers put forward schemes for obfuscating general computing \cite{2013Garg}. A parallel line of research over the last decade has demonstrated that most cryptographic primitives, including several powerful ones such as multiparty non-interactive key exchange \cite{BZ14}, succinct non-interactive arguments \cite{SW21}, software watermarking \cite{2016stoc},  and deniable encryption \cite{SW21} can be built from indistinguishability obfuscation (iO). Moreover, iO has also found appeal outside cryptography, e.g. for establishing the hardness of certain tasks in differential privacy \cite{BZ14}. These results have established iO as a `central hub' of theoretical cryptography. 
\begin{table*}[ht]
\caption{Related Researchs on Quantum Indistinguishability Obfuscation}
\resizebox{\textwidth}{!}{
    \centering
    \begin{threeparttable}
    \renewcommand{\arraystretch}{1.5}
    \setlength{\tabcolsep}{10pt}
    \begin{tabular}{||c|c|c|c|c|c||}
    \hline
    \textbf{Reference} & \makecell[c]{\textbf{General} \\\textbf{ purpose} } & \textbf{Function} & \textbf{Ability of challenger or verifier} & \textbf{Security loss} & \textbf{Year} \\ \hline
    \hline
    [10] & \ding{56} & Quantum circuits obfuscation with logarithmic non-Clifford gates & Classical computation under LWE assumption & Exponential & 2021 \\ \hline
    [11] & \ding{56} & Trivial quantum circuits obfuscation with fixed zero outputs & Classical computation under LWE assumption & Exponential & 2022 \\ \hline
    [18] & \ding{52} & \makecell[c]{Fixed universal quantum circuits obfuscation \\ and assume that the description of $C$ is included in $\rho$ }& Polynomial-size quantum computation & Exponential & 2024 \\ \hline
    \textbf{Our Work} & \ding{52} & Universal quantum circuits obfuscation & Polynomial-size quantum computation & Polynomial & - \\ \hline
    \end{tabular}
    
    \begin{tablenotes}
    \item[*] $C$ and $\rho$ are circuit-state pairs of quantum implementation of classical function
    \end{tablenotes}
    \end{threeparttable}
    }
    
    \label{tab:qio_comparison}
\end{table*}
Similar to classical obfuscation, quantum obfuscation provides a viable technological solution to guaranteeing secure delegated quantum computing while protecting the execution logic of commercially sensitive quantum programs from potential adversary. The characteristics of quantum entanglement and superposition hinder direct application of classical obfuscation in quantum obfuscation. Research on quantum obfuscation has predominantly focused on definitional work \cite{AF16}, impossibility results \cite{AF16,ABDS21}, and limited classes of quantum computing \cite{Shang19,Zhang22,BK21,BM22}. 

One prominent framework for quantum obfuscation is Virtual Black Box (VBB) obfuscation. The idea behind VBB is that an obfuscated program should behave like a black box when observed by any adversarial party, ensuring that regardless of the adversary’s observations, analyses, or queries, they cannot extract more information than what a direct query to the original program’s oracle would reveal. However, achieving VBB obfuscation for general quantum circuits has proven challenging. Alagic et al. \cite{ABDS21} demonstrated the provable existence of circuits that cannot be VBB obfuscation, even when utilizing quantum information. 

Therefore, subsequent research on quantum obfuscation has primarily adopted two approaches, one involves constructing VBB obfuscation for limited classes of quantum computing, and the other focuses on relaxation of VBB obfuscation towards general-purpose. Recent research has explored VBB obfuscation under certain restricted scenarios, such as quantum point function obfuscation \cite{Shang19,Zhang21} and its instantiation schemes \cite{Zhang22}, black-box constructions for unlearnable programs \cite{ALL+21}, polynomial-size pseudo-deterministic quantum circuit \cite{2023Pseudo}, quantum non-linear functions \cite{Pan2023}, have demonstrated its viability under specific conditions.

Unlike the limited classes of quantum computing for VBB obfuscation research above, Barak et al. \cite{BGI+12} introduced a relaxation concept of obfuscation known as indistinguishability obfuscation, which only requires that the obfuscation of two functionally equivalent programs is indistinguishable in computing. Despite this relaxation, subsequent works such as \cite{GGH+16,SW21} demonstrated its considerable strength. Quantum indistinguishability obfuscation (QiO) possesses the potential to construct various other quantum cryptographic primitives and can be applied to a wider range of quantum computing tasks and application scenarios.

However, as shown in Table ~\ref{tab:qio_comparison}, designing QiO schemes for general-purpose quantum computing remains challenges. As far as we know about the existing constructions from the literature, such as those by Broadbent et al. \cite{BK21}, Bartusek et al. \cite{BM22}, and Coladangelo et al. \cite{2023How}, constraint QiO to restricted classes of quantum circuits, such as those using limited non-Clifford gates, trivial output functions, or a fixed universal quantum circuit for best-possible copy protection. Instead of considering implementations as circuit-state pairs, the description of circuit is assumed to be included in state. Despite their theoretical importance, these constructions fall short of addressing the complexities inherent in general quantum circuits. Additionally, proving functional equivalence for general quantum circuits often requires exhaustive enumeration of all possible inputs. This leads to exponential security loss, making such schemes impractical for large-scale or complex circuits. Even the best current constructions rely on inefficient exhaustive verification, which scales poorly and restricts their applicability to limited scenarios.

In this paper, we propose a scheme aimed at constructing QiO by leveraging the equivalence proof of quantum circuits. The starting point is a simple observation: assuming the existence of a secure quantum indistinguishability obfuscator, a mathematical proof is developed to convince oneself and others that two programs are functionally equivalent. We introduce a polynomial-time probabilistic test for functional equivalence, significantly mitigating the exponential security loss inherent in previous approaches. By hybrid argument, we demonstrate that two quantum implementations of classical functions are quantum indistinguishable.
The main contributions of this paper is described as follows:
\begin{enumerate}
\item We overcome the constraints of existing QiO schemes that apply only to certain classes of quantum circuits. Through semantic transformations of quantum circuits, we construct a general-purpose QiO scheme that works for arbitrary quantum circuits, bypassing the need for restrictive gate classes and trivial output functions. It also eliminates the restrictive assumptions of previous QiO constructions and solves the problem of exponential security loss.

\item We propose QiO scheme via quantum circuit equivalence (QceQiO). It is the candidate construction of indistinguishable obfuscation for general quantum circuits. It relies on the reverse engineering of circuit equivalence through succinct mathematical proofs, enabling efficient and scalable obfuscation for arbitrary quantum circuits.

\item 
We introduce the concept of incremental subpath equivalence, which involves transforming quantum circuits through incremental changes in subpaths. Two circuits only differ in two polynomial-sized functionally equivalent subpaths at the same position, and each incremental change is of small scale.
In addition, the problem of distinguishing quantum circuits is translated into a polynomial identity testing problem.


\end{enumerate}


\section{Preliminaries}
\label{section2}
\subsection{Quantum Circuits}
In the classical circuit model, the state of an $n$-bit system is represented as a binary string of length $n$. Classical gates correspond to operators that map binary strings of length $n$ to binary strings of length $m$. More precisely, a binary string of length $n$ can be viewed as a vector in $\mathbb{F}_{2}^{n}$, where $\mathbb{F}_{2}$ is the finite field with addition corresponding to Boolean XOR $(\oplus)$ and multiplication corresponding to Boolean AND $(\land)$. Classical gates are then represented as operators $\left. f:\mathbb{F}_{2}^{n}\rightarrow\mathbb{F}_{2}^{m} \right.$, and we often refer to $f$ as a classical function. 

The quantum circuit model is an extension of the classical circuit model to handle quantum effects in quantum computation. Specifically, it describes the state space of an $n$-qubit system as a vector in a $2^n$-dimensional complex vector space $\mathcal{H}$, spanned by the $n$-qubit states. Conventionally, classical states are referred to as the standard or computational basis of $\mathcal{H}$, and they are represented by Dirac notation: $\left. \left| x \right. \right\rangle$,where $x \in \mathbb{F}_{2}^{n}$. 
More generally, a quantum system can be described as a probability distribution over a set of computational basis states denoted as $\left| x_j \right\rangle$, where each $\left| x_j \right\rangle$ represents a particular state in the state space. Note that these states are not the standard basis vectors but instead correspond to a specific set of quantum states that are meaningful in the context of the given system. The probability of the system being in the state $\left| x_j \right\rangle$ is denoted by $p_j$. The density operator $\rho$, which represents the overall quantum state of the system, is defined as:$\rho: = {\sum\limits_{j}{p_{j}\left. \left| x_{j} \right. \right\rangle}}\left\langle \left. x_{j} \right| \right.$, where $p_j$ is the probability associated with each computational basis state $\left| x_{j} \right\rangle$.
For a quantum system with $n$ qubits, its density operator is a Hermitian matrix, which is a complex matrix. 

For a closed physical system, its quantum state evolves through a unitary transformation. Specifically, assuming the quantum state of the closed system at any two time instances $t_1$ and $t_2$ is represented by $\rho$ and $\rho^{\prime}$ respectively, the evolution between $\rho$ and $\rho^{\prime}$ is governed by a unitary transformation $U$, determined by the time interval $[t_1,t_2]$. Mathematically, this transformation is expressed as 
$\rho^{\prime} = \left. \left| Ux \right. \right\rangle\left\langle \left. Ux \right| \right. = U\left. \left| x \right. \right\rangle\left\langle \left. x \right| \right.U^{\dagger} = U\rho U^{\dagger}.$
\vspace{-10pt}
\begin{figure}[h]
  \centering
  \subfloat[Identity gate]{\includegraphics[width=0.25\linewidth]{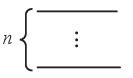}\label{fig:a}}\hspace{0.05\linewidth}
  \subfloat[$n$-qubit gate]{\includegraphics[width=0.3\linewidth]{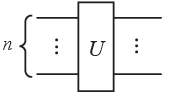}\label{fig:b}}
  \caption{Identity gate $I_n$ and $n$-qubits gate $U$}
\end{figure}

In graphical representation, the $n$-qubits gate that does not apply any operation to qubits is known as an identity gate, typically denoted by $I_n$, as shown in Fig. 1.(a). The $n$-qubits gate is depicted as a box $U$ with $n$ input and output lines, as shown in Fig. 1.(b).
In a quantum circuit, complex quantum circuits can be constructed by combining $I$ with different $U$ gates both horizontally and vertically. 
Through these combinations, any quantum computing process can be described.
\subsection{Quantum Indistinguishable Obfuscation}
Quantum obfuscation is a quantum derivative of classical obfuscation theory, developed at the intersection of quantum computing and quantum cryptography. It explores efficient ways to achieve functional encryption and copyright protection within quantum circuits.  Obfuscation was defined by three crucial properties, namely functional equivalence, polynomial expansion, and virtual black box security, was first formulated by Barak et al. \cite{Barak2012}. Alagic et al. \cite{AF16} introduced the concept of classical obfuscation to the quantum domain, defined three corresponding properties for quantum obfuscation. It takes a quantum circuit $C$ as input and produces a `functionally equivalent' quantum state $O(C)$ as output.

\begin{definition}
\label{def_qvbbo}
(Quantum Virtual Black Box Obfuscation): Quantum VBB Obfuscation consists of a quantum algorithm $O$ and a quantum polynomial-time algorithm $\delta$. For any quantum circuit $C$ with $n$ qubits, the quantum algorithm $O$ outputs a quantum state $O(C)$ with $m$ qubits, satisfying the following three conditions:
\begin{itemize}
  \item Polynomial expansion: $m = poly(n)$
  \item Functional equivalence: For any quantum state $\rho$ with $n$ qubits, $
\left\| {\delta\left( {O(C)\bigotimes\rho} \right) - U_{C}\rho U_{C}^{\dagger}} \right\|_{tr} \leq negl(n).$
where $U_C$ represents the unitary matrix corresponding to $C$.
  \item Virtual black box: For any quantum polynomial-time algorithm A, there exists a quantum polynomial-time algorithm $S^{U_{C}}$ such that: $
\left| {{\Pr\left\lbrack {A\left( {O(C)} \right) = 1} \right\rbrack} - \Pr\left\lbrack S^{U_{C}}\left( \left. \left| 0^{n} \right. \right\rangle \right) = 1 \right\rbrack} \right| \leq negl(n).$
\end{itemize}

\end{definition}

Alagic et al. \cite{ABDS21} proved the nonexistence of general Quantum VBB Obfuscation. Consequently, similar to the relaxation in classical obfuscation, research in quantum obfuscation has been divided into two approaches, namely VBB obfuscation for specific functions and relaxation of VBB obfuscation. To weaken the virtual black box property, Broadbent et al. \cite{BK21} introduced the definition of indistinguishability of quantum states.

\begin{definition}
\label{def_indistingqs}
(Indistinguishability of Quantum States): For two sets of quantum states, $\mathcal{R} = \left\{ \rho_{n} \right\}_{n \in \mathbb{N}}$ and $\mathcal{S} = \left\{ \sigma_{n} \right\}_{n \in \mathbb{N}}$, where $\rho_{n},\sigma_{n}$ are quantum states,
\begin{itemize}
    \item If, for any $n$, $\rho_{n} = \sigma_{n}$, then $R$ and $S$ are perfectly indistinguishable.
    \item If there exists a negligible function negl, such that for all sufficiently large $n$, $\left\| {\rho_{n} - \sigma_{n}} \right\|_{F}   \leq negl(n)$, then $R$ and $S$ are statistically indistinguishable.
    \item If there exists a negligible function negl, such that for every quantum state $\rho_{n} \in \mathcal{R},\sigma_{n} \in \mathcal{S}$, and all polynomial-time quantum distinguishers $D$, $\left| {{\Pr\left\lbrack {D\left( \rho_{n} \right) = 1} \right\rbrack} - {\Pr\left\lbrack {D\left( \sigma_{n} \right) = 1} \right\rbrack}} \right| \leq negl(n),$ then $R$ and $S$ are computationally indistinguishable.
\end{itemize}
\end{definition}


In practical equivalence testing, a major and often overlooked concern involves numerical inaccuracies. Matrices are typically stored by floating-point numbers, introducing imprecision and rounding errors. 
Generally, statistical indistinguishability is considered a stronger security concept than computational indistinguishability, as it provides an information-theoretic guarantee that no algorithm can distinguish between two distributions, regardless of computational resources. In contrast, computational indistinguishability implies that no algorithm can efficiently distinguish in polynomial time, which is less robust when compared to statistical indistinguishability. In some contexts, such as discussions of cryptography and security, greater emphasis is placed on computational indistinguishability as it corresponds more closely with the capabilities of real attackers.


\subsection{Quantum Path Integral}
\label{qpathinte}
\begin{figure}[h]
  \centering
  \subfloat{\includegraphics[width=0.43\linewidth]{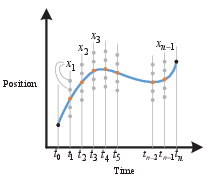}\label{fig:a}}\hspace{0.05\linewidth}
  \subfloat{\includegraphics[width=0.43\linewidth]{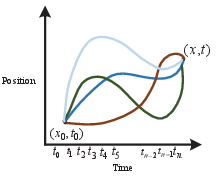}\label{fig:b}}
\label{fig3}
  \caption{Quantum Paths and Fluctuations in the Feynman Path Integral. In Feynman path integral method, the state of a physical system can be described by considering all possible paths. 
These paths include classical trajectories as well as various possibilities arising from quantum fluctuations.}
\end{figure}
The path integral formulation of quantum mechanics was formalized by Feynman. It provides a generalization and formalization of quantum physics extending from the principle of classical mechanics. The path integral method expresses the evolution of the wave function after a finite time through the integral transformation of the propagator. The connection between the quantum states $\psi(t)$ and $\psi\left( t_{0} \right)$ is given by the evolution operator $U(t, t_0)$,
$\left. \left| \psi(t) \right. \right\rangle = U\left( t,t_{0} \right)\left. \left| \psi\left( t_{0} \right) \right. \right\rangle$.


According to the Schr$\ddot{o}$dinger equation $i\hbar \frac{d}{dt}\left. \left| \psi(t) \right. \right\rangle = H\left. \left| \psi(t) \right. \right\rangle$, the evolution operator $U$ is defined as,
$U\left( {t,t_{0}} \right) = e^{- i\frac{H}{\hbar}(t - t_{0})}$
, where $\hbar$ is the reduced Planck constant, $H$ is the Hamiltonian operator, and the coordinate representation of the evolution operator is the propagator,
$K\left( {x,t;x_{0},t_{0}} \right) = \left\langle x \middle| {U\left( t,t_{0} \right)} \middle| x_{0} \right\rangle.$

The meaning of the propagator is that it represents the probability amplitude for a particle to be at position $x$ at time $t$ when it was at position $x_0$ at the time $t_0$. The evolution operator satisfies the composition rule $U\left( {t_{1},t_{0}} \right) = U\left( {t_{1},t} \right)U\left( {t,t_{0}} \right)$, and correspondingly, the propagator has the expression,
$K\left( {x_{1},t_{1};x_{0},t_{0}} \right) = {\int_{\Lambda}^{}{\mathcal{D}\lbrack P\rbrack e^{iS(P)/h}}}.$


Here $\Lambda$ represents all paths from $(x_0, t_0)$ to $(x, t)$, and $S(P)$ is the action of the particle along the path $P$ from $x_0$ to $x$. $S(P)$ is the Hamiltonian action function for the path $P$. The contribution of the path $P$ depends on $e^{iS(P)/\hbar}$, where $\hbar = 1$. $\mathcal{D}$ represents functional differentiation, indicating integral over all possible path $P$.

\begin{figure*}[htb]
\centering
\includegraphics[width=0.95\linewidth]{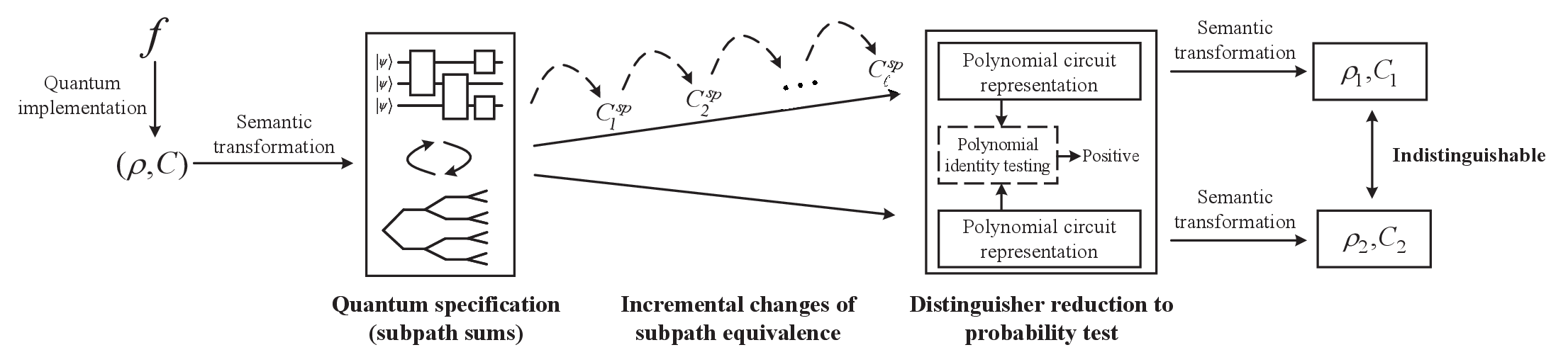}
\caption{Framework of quantum indistinguishable obfuscation}
\label{fig13}
\end{figure*}

\section{Quantum Indistinguishable Obfuscation Scheme via Quantum Circuit Equivalence}
\label{section3}
In this section, we provide the definition and construction of quantum indistinguishability obfuscation scheme via quantum circuit equivalence (QceQiO).

\subsection{Quantum Indistinguishability Obfuscation Scheme}
A framework of quantum indistinguishable obfuscation scheme is shown in Fig. \ref{fig13}. A function $f$ take semantic transformation to generate a quantum specification comprising the quantum implementation. This specification is then represented as subpath sums, forming the basis for subsequent transformations. By introducing incremental changes in subpath equivalence, QceQiO progressively constructs obfuscated circuits, preserving functional equivalence while ensuring indistinguishability. Polynomial identity testing is applied to the phase polynomials of quantum subpaths, reducing the distinguisher to a polynomial probability test.

Before presenting quantum indistinguishability obfuscation for equivalent quantum circuits, we provide the definition of an equivalent quantum implementation of a function $f$.
The quantum implementation of a classical function is defined as follows \cite{2023How}:
\begin{definition}{(Quantum Implementation of Classical Function):}
Let $n,m \in \mathbb{N}$, $\left. f:\left\{ 0,~1 \right\}^{n}\rightarrow\left\{ 0,~1 \right\}^{m} \right.$, and $\epsilon \in \lbrack 0,~1\rbrack$. A $(1-\epsilon)$-quantum implementation of $f$ is a pair $(\rho,C)$, where $\rho$ is an initial quantum state,  $C$ is a quantum circuit acting on $\rho$ and an additional $n$-qubit register initialized with the classical input $x$. The implementation satisfies:
\begin{equation}
\forall x \in \left\{ {0,~1} \right\}^{n},\Pr \left\lbrack C(\rho,~x) = f(x) \right\rbrack \geq 1 - \epsilon.
\end{equation}
where for the quantum implementation $(\rho,C)$, its size is defined as the maximum of the number of qubits in $\rho$ and the number of gates in the circuit $C$.
where  $C(\rho,x)$ denotes the output state after applying the quantum circuit $C$ on the input state  $(\rho,x)$, and $\Pr$ refers to the probability over the randomness in the measurement outcomes. The size of the quantum implementation is defined as the maximum of the number of qubits in $\rho$ and the number of gates in $C$.
\label{def_qclassf}
\end{definition}

In contrast to the classical circuit model, quantum gates in a quantum circuit are restricted to a subset of all operators on $\mathcal{H}$. Specifically, quantum gates are linear operators $\left. U:\mathcal{H}\rightarrow\mathcal{H} \right.$ that preserve the $L_2$ norm. Such operators $U$ satisfy $U^{\dagger}U = UU^{\dagger} = ~I$, where $U^{\dagger}$ denotes the adjoint of $U$, and they are referred to as unitary operators. Quantum circuits essentially describe the action of unitary operators on a Hilbert space.

\begin{definition}
(Equivalent Quantum Implementation of $f$): Let $n,m \in \mathbb{N}$, $\left. f:\left\{ 0,~1 \right\}^{n}\rightarrow\left\{ 0,~1 \right\}^{m} \right.$, $\left( \rho_{0},C_{0} \right)$ and $\left( \rho_{1},C_{1} \right)$ be two equivalent quantum implementations of $f$,
\begin{equation}
\left| {{\Pr\left\lbrack {D\left( {\rho_{0},C_{0}} \right) = 1} \right\rbrack} - {\Pr\left\lbrack {D\left( {\rho_{1},C_{1}} \right) = 1} \right\rbrack}} \right| \leq negl(\lambda)
\end{equation}
where $C$ is a quantum circuit satisfying $\forall x \in \left\{ {0,~1} \right\}^{n},\Pr\left\lbrack C(\rho,~x) = f(x) \right\rbrack \geq 1 - negl(\lambda)$.
\end{definition}
In quantum computing, equivalent quantum implementations can perform the same quantum computing functionality, which means that copy protection generates outputs will be indistinguishable when equivalent quantum implementations are input. This is a fundamental approach for copy protection of quantum implementations. In recent work of Andrea et al. \cite{2023How}, the approach is also referred to as the best possible copy protection. We present the definition of QceQiO.

\begin{definition}
\label{qceQiO}
(Quantum Indistinguishability Obfuscation via Quantum Circuits  Equivalence-QceQiO): Let $\left\{ Q_{\lambda} \right\}_{\lambda \in \mathbb{N}}$ be a family of quantum implementations for the classical function $f$, and $C$ be a family of quantum circuits. A quantum indistinguishability obfuscator for equivalent quantum circuits is a quantum polynomial-time (QPT) algorithm QceQiO that takes as input a security parameter $1^{\lambda}$ and a pair of quantum implementations $(\rho,C) \in Q_{\lambda}$, and outputs a pair of $\left( \rho^{\prime},C^{\prime} \right)$. Additionally, QceQiO should satisfy the following conditions:
\begin{itemize}
\item Polynomial Expansion: There exists a polynomial function $\text{poly}(n)$ such that for all $C \in \mathcal{C}$, $\mathcal{C}$ is quantum circuit family ,the size of the obfuscated circuit $C^{\prime}$ satisfies:
\begin{equation}
|C'| = \text{poly}(|C|).
\end{equation}
This means that the size of the obfuscated circuit $C^{\prime}$ is polynomially bounded in terms of the size of $C$.
\item Functional equivalence: For any $C \in \mathcal{C}$, 
\begin{gather}
(\rho', C') \leftarrow QceQiO(\rho, C) \text{,} \nonumber\\
 C \text{ and } C' \text{ are under } \Delta \text{ subpath equivalence.}
\end{gather}
\item Computational indistinguishability: For any QPT distinguisher $D$, there exists a negligible function $negl$ such that for all $\lambda$ and two pairs of quantum implementations $(\rho_{1}, C_{1}), (\rho_{2}, C_{2})$ of the same function $f$, the distributions of the obfuscated outputs are computationally indistinguishable:
\begin{equation}
\resizebox{0.47\textwidth}{!}{
$\left| \begin{matrix}
{{\Pr\left\lbrack {D\left( ~QceQiO\left( {1^{\lambda},~\left( {\rho_1,C_1} \right)} \right)\rightarrow\left( {\rho_{1}^{\prime},C_{1}^{\prime}} \right) \right) = 1} \right\rbrack} -} \\
{\Pr\left\lbrack {D\left( QceQiO\left( {1^{\lambda},~\left( {\rho_2,C_2} \right)} \right)\rightarrow\left( {\rho_{2}^{\prime},C_{2}^{\prime}} \right) \right) = 1} \right\rbrack}
\end{matrix} \right| \leq ~negl(\lambda).$}
\end{equation}
\end{itemize}
\end{definition}
Similar to the indistinguishability of quantum states represented by density operators, distinguishing the equivalence of two quantum implementations of classical function $f$ can be simplified. Simplification involves sequentially evolving all inputs using the constructed unitary matrices. This is also an implicit strategy in the security proofs of all universal iO structures. Although conceptually straightforward, it quickly becomes a challenging task due to the exponential scaling of matrices with the number of qubits \cite{1993bern}. It is also an implicit security loss in all indistinguishability obfuscation.

As seen from the definition of QiO, if two quantum circuits $C_1$ and $C_2$ exhibit functional equivalence, i.e., for all inputs $x$, $C_{1}(x) = C_{2}(x)$, then their obfuscations must be indistinguishable for any polynomial-time algorithm. Assuming there exists a polynomial-time reduction that, given access to the oracle of a QiO scheme, breaks the assumed security based on some polynomial-time hardness assumption. Consider a `trivial' polynomial-time adversary $A$ that selects two quantum circuits $C_{1}(x) = C_{2}(x)$, which are functionally equivalent except for differing outputs at some input $x^{*}$. This adversary can distinguish the obfuscations of $C_1$ and $C_2$ by evaluating them at $x^{*}$. To prevent from being fooled by adversary $A$, it is necessary to check whether the two quantum circuits $C_1$ and $C_2$ are functionally equivalent for all inputs. A natural approach is to traverse all inputs one by one, which is also the implicit strategy in the security proofs of all existing generic QiO constructions. Consider a simple scenario where the outputs of two quantum circuits $C_1$ and $C_2$ differ only at a specific $x^{*}$ for all inputs. If the input space is two-dimensional, traversing all possible inputs one by one requires checking $2^n$ different inputs, where $n$ is the input length, leading to exponential computational complexity. This strategy has also been named by Jain et al. \cite{2022jain} as suffering from the security loss with input length increases exponentially. Therefore, in the design of QiO, checking functional equivalence is necessary. Without such verification, it is impossible to ensure consistent obfuscation in terms of functionality.

\subsection{Quantum Specification Based on Subpath Sum}
We employ the conventional notation for paths and symbols \cite{Amy2018}, which can be viewed as a compact representation of the state transition relations. Since the evolution operator $U$ is typically modeled as an operator on the finite-dimensional Hilbert space $\mathcal{H}$, we use discrete sums instead of integrals, aligning with recent research on quantum path semantics \cite{Montanaro_2017,2023Synthesizing}. Specifically, instead of providing a separate path integral for each transition amplitude, we use a set of paths (denoted by the symbol $\Lambda$) connecting the basis states, paired with a set of amplitudes $(\left. \phi:\Lambda\rightarrow\mathbb{~}\mathbb{C} \right.)$, where $\mathbb{C}$ represents the complex field. This set of paths and amplitudes naturally corresponds to the operator $U = {\sum\limits_{\tau:x_{0}\rightarrow x \in \Lambda}{\phi(\tau)\left. \left| x \right. \right\rangle\left\langle \left. x_{0} \right| \right.}}.$

Here $\left. \tau:x_{0}\rightarrow x \right.$ denotes that $\tau$ is a path from $\left| x_{0} \right\rangle$ to $\left| x \right\rangle$. Any unitary matrix $U \in U\left( 2^{n} \right)$ can also be described as a sum-over-paths action by including a path $\left. \Lambda_{x,~x^{\prime}}:x\rightarrow x^{\prime} \right.$ in $\Lambda$ for every pair $x,~x^{\prime} \in \mathbb{F}_{2}^{n}$, where $\phi\left( \Lambda_{x,~x^{\prime}} \right) = \left\langle x^{\prime} \middle| U \middle| x \right\rangle$. The operator $U$ can be represented as an action on the basis states,
\resizebox{0.48\textwidth}{!}{$\left. U:\left. \left| x_{0} \right. \right\rangle\mapsto{\sum\limits_{x \in \mathbb{F}_{2}^{n}}{\left( \left\langle x \middle| U \middle| x_{0} \right\rangle \right)\left. \left| x \right. \right\rangle}}\mapsto{\sum\limits_{\tau:x_{0}\rightarrow x \in \Lambda_{x_{0}}}{\phi(\tau)\left. \left| x \right. \right\rangle}} \right.$}.

It is referred to as a sum-over-paths action defined by $\Lambda$ and $\phi$, where $\Lambda_{x_{0}} \subseteq \Lambda$ represents the paths starting from $\left| x_{0} \right\rangle$.

In a general sense, the idea involves considering the computational basis states as physical states of a quantum bit system. A unitary operator maps a specific state to a superposition of a series of states, each with a specific probability amplitude. The mapping process can be viewed as the evolution of an initial state (possibly in a superposition) to other states through a combination of paths.

The path of $X$, $H$, and $R_z$ gates are represented as follows:
$
\left. X:\left. \left| x \right. \right\rangle\rightarrow\left. \left| 1 \oplus x \right. \right\rangle, \right.
\left. H:\left. \left| x \right. \right\rangle\rightarrow\frac{1}{\sqrt{2}}{\sum\limits_{y \in {\{ 0,1\}}}e^{i\pi xy}}\left. \left| y \right. \right\rangle, \right.\nonumber\\$
\resizebox{0.24\textwidth}{!}{$
\left. ~R_{z}^{\theta}:\left. \left| x \right. \right\rangle\rightarrow e^{i{({2x - 1})}\theta}\left. \left| x \right. \right\rangle,~ \right. $}
\resizebox{0.23\textwidth}{!}{$
\left. \left. CX:~ \middle| x_{1}x_{2} \right\rangle\rightarrow \middle| x_{1}~\left( x_{1} \oplus x_{2} \right)\rangle \right.
$}

Applying the $X$ gate to the basis state $\left| x \right\rangle$ results in the state $\left| \neg x \right\rangle$. Applying $H$ to $\left| \neg x \right\rangle$ results in the state $\frac{1}{\sqrt{2}} |0\rangle + \frac{1}{\sqrt{2}} e^{i\pi x} |1\rangle$. The $R_z$ gate, parameterized by the angle $\theta$, only changes the amplitude of a given basis state, as seen in $e^{i(2x - 1)\theta}$, where $2x-1$ is actually a coefficient related to $x$. Here $x$ represents the bit value in the computational basis $\left| x \right\rangle$, which can be 0 or 1. The value of $2x-1=-1$ when $x=0$ and $2x-1=1$ when $x=1$. Its value is related to the direction of rotation of the qubit around the $z$-axis. The controlled $X$ $(CX)$ gate can entangle two qubits, a crucial operation in quantum computing: $\left. \left| x_{1}x_{2} \right\rangle\rightarrow \middle| x_{1}~\left( x_{1} \oplus x_{2} \right)\rangle \right.$ Given the basis state $\left| x_{1}x_{2} \right\rangle$, $CX$ generates the basis state $\left| x_{1}\left( x_{1} \oplus x_{2} \right) \right\rangle$, where $\oplus$ is the XOR operation.

The path sum computation represents each quantum gate in the quantum circuit $C$ by acting on the computational basis states. It consists of deterministic output Boolean functions and semi-Boolean functions that determine the relative phases of the outputs, with each output depending on the input and `path'. A semi-Boolean function is defined as a mapping from a binary input set $\{0,1\}^n$ to the set of complex numbers, capturing the interaction between amplitude and phase for specific quantum states. We divide the time interval $[t_a, t_b]$ over which the quantum circuit $C$ evolves into $N$ grid points, where the specific division is given by $\frac{t_{b} - t_{a}}{N} = \Delta t,t_{m} \coloneqq t_{a} + m \Delta t$, resulting in a series of time steps $t_{a} = t_{0} < t_{1} < \cdots < t_{N - 1} < t_{N} = t_{b}$.

Subpath sum origins from the description of quantum circuit functional verification by Feynman path integral \cite{Amy2018}, and we connect path sum and quantum Fourier expansion, representing the action of relevant operators by informally expressing path sums. The discrete fourier transform is applied to the quantum state amplitudes within each time step $\mathrm{\Delta}t$ to represent the action on the computational basis states.
\begin{equation}
\left. \left. \left| x \right. \right\rangle\rightarrow\frac{1}{\sqrt{2^{n}}}{\sum\limits_{y = 0}^{2^{n} - 1}e^{2\pi ixy/2^{n}}}\left. \left| y \right. \right\rangle \right.
\end{equation}

Fig. \ref{figqft} illustrates a 3-qubits quantum Fourier transform circuit. The controlled-phase rotation gates are typically represented as Control-$R_k$ gates, $R_{k} = \begin{bmatrix}
1 & 0 \\
0 & e^{2\pi i/2^{k}}
\end{bmatrix}$, where $k$ denotes the angle of phase rotation. Different values of $k$ correspond to different phase components.

\begin{figure}[h]

\hspace{1.5cm} 
\resizebox{0.5\linewidth}{!}{
\Qcircuit @C=0.5em @R=0.7em {
\lstick{\ket{j_i}} & \gate{H} & \ctrl{1}  & \ctrl{2} &  \qw   &  \qw        &\qw&\rstick{\ket{0}+e^{0.j_1j_2j_3}\ket{1}} \qw \\
\lstick{\ket{j_2}} &     \qw  & \gate{R_2}& \qw    &    \gate{H}  & \ctrl{1} & \qw&\rstick{\ket{0}+e^{0.j_2j_3}\ket{1}} \qw \\
\lstick{\ket{j_3}} &     \qw  &  \qw      & \gate{R_3}  &  \qw     &\gate{R_2} &\gate{H}&\rstick{\ket{0}+e^{0.j_3}\ket{1}} \qw 
}}
\caption{3-qubits quantum Fourier transform circuit}
\label{figqft}
\end{figure}
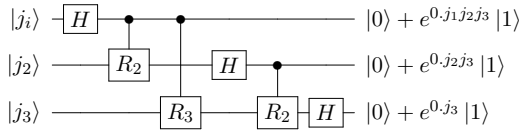


\begin{definition}
\label{def_SP}
(Subpath Sum \( SP \)): The subpath sum \( SP \) of a quantum circuit \( C \) at time \( t_m \) is the set of exponentiated sums of the basis states at time \( t_m \). The domain of \( SP \) consists of computational basis states parameterized by input variables \( x \) and path variables \( y \), where each basis state represents a possible computational path in the circuit.
\begin{itemize}
  \item Input basis vector: \( \left| x = x_{1}x_{2}\cdots x_{n} \right\rangle \), where each \(x_i\) is a Boolean variable representing the \(i\)-th input qubit in \( C \), and \(n\) is the number of input qubits.
  \item Phase polynomial: \( \phi \in D_{M}[x,y] \) is a linear polynomial over the input variables \( x \) and path variables \( y = y_{1}y_{2}\cdots y_{m} \), where \( m \) is number of intermediate path variables in circuit. Phase polynomial encodes relative phases accumulated along each computational path.
  \item Output basis vector: \( \left| f(x,y) = f_{1}(x,y)\cdots f_{n}(x,y) \right\rangle \), where each \( f_i \in \mathbb{Z}_2[x,y] \) is a Boolean polynomial representing the \(i\)-th output qubit as a function of the inputs and path variables.
  \item Functionality:The subpath sum \( SP \) is generated by the unitary evolution of the quantum circuit \(C\) from the input basis state \( \left| x \right\rangle \) to the output basis state \( \left| f(x,y) \right\rangle \). The corresponding unitary operator \( U_{SP} \) is defined as: $\left| x \right\rangle \rightarrow \frac{1}{\sqrt{2^m}} \sum_{y \in \mathbb{Z}_2^m} e^{2\pi i \phi(x,y)} \left| f(x,y) \right\rangle$

\end{itemize}
\end{definition}

Some subpath sums within their evolution $\Delta t$ do not transform the basis state into a superposition state because they represent monomial unitary operators. The subpath sum can be simplified to be $\left. \left| x \right\rangle\rightarrow\phi(x) \middle| f(x)\rangle \right.$

For example, the unitary operator representation of the CX gate in terms of single-qubit states is given by $\left. CX = \middle| 0 \right\rangle\left\langle 0 \middle| \otimes I + \middle| 1 \right\rangle\left\langle 1 \middle| \otimes X \right.$

If each element of a unitary operator can be expressed as the product of a single-qubit state and a coefficient, then the unitary operator is a monomial unitary operator. Here $\otimes$ denotes the tensor product, $I$ is the identity unitary for a single qubit, and $X$ is the Pauli-$X$ gate. The subpath sum representation of the $CX$ gate is given by $\left. \left. \left| x_{1}x_{2} \right. \right\rangle\rightarrow\left. \left| x_{1}\left( x_{1}\bigoplus x_{2} \right) \right. \right\rangle \right.$. Thus, $\phi = 1$, there are no path variables $y$, and $f = \left( f_{1},f_{2} \right)$, where $f_{1}\left( {x_{1},x_{2}} \right) = x_{1}$ and $f_{2}\left( {x_{1},x_{2}} \right) = x_{1}\bigoplus x_{2}$.

The quantum circuit, also known as quantum logic circuit, is the most commonly used model for universal quantum computing. It represents, in abstract terms, a circuit that operates on quantum qubits.
As a quantum process, the operation of a quantum circuit can also be described as the sum of all transitions between intermediate states.
%
\begin{figure}[h]
\centering
\includegraphics[width=0.95\linewidth]{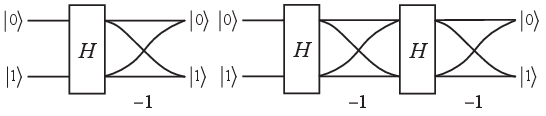}
\caption{Schematic diagram illustrating the paths taken by each input state and the phases acquired along each subpath in $HH$ circuit}
\label{fig5}
\end{figure}
\begin{theorem}
\label{theo_qlogic}
(Quantum Circuit Path Representation): Let \( C \) be a quantum circuit composed of \( n \) input qubits and \( m \) intermediate qubits represented by path variables \( y = y_1 y_2 \cdots y_m \). The output amplitude of \( C \) for a given input state \( |x\rangle \) and output state \( |f(x, y)\rangle \) is given by:
\begin{equation}
\left\langle f(x, y) \middle| C \middle| x \right\rangle = \frac{1}{\sqrt{2^m}} \sum_{y \in \mathbb{Z}_2^m} \alpha(x, y) e^{2\pi i \phi(x, y)},
\end{equation}
where \( \phi(x, y) \) is the phase polynomial encoding the phase contribution of each path.
\end{theorem}

\begin{proof}
Let \( C \) be a quantum circuit with \( n \) input qubits and \( m \) intermediate path variables \( y = y_1 y_2 \cdots y_m \). For a given input state \( |x\rangle \), the evolution of \( C \) can be described as a sum over all computational paths indexed by \( y \). Each path \( y \) accumulates a phase contribution encoded by the polynomial \( \phi(x, y) \), giving the overall amplitude: $\left\langle f(x, y) \middle| C \middle| x \right\rangle = \frac{1}{\sqrt{2^m}} \sum_{y \in \mathbb{Z}_2^m} \alpha(x, y) e^{2\pi i \phi(x, y)}.$
Here, \( \alpha(x, y) \) represents the amplitude of each path (typically \( 1 \) for basic Clifford circuits), and \( \phi(x, y) \) is a polynomial function of \( x \) and \( y \).

To understand interference effects, we analyze the specific example of a single-qubit Hadamard gate:
$H |x\rangle = \frac{1}{\sqrt{2}} \sum_{y \in \{0, 1\}} e^{i\pi xy} |y\rangle.$
The phase polynomial is \( \phi(x, y) = \pi x y \), where \( y \) takes values in \( \{0, 1\} \). When \( x = 0 \), the phase \( \phi(0, y) = 0 \) for all \( y \), leading to constructive interference:
$\langle 0 | H H | 0 \rangle = \frac{1}{\sqrt{2}} \cdot \frac{1}{\sqrt{2}} (1 + 1) = \frac{1}{2}(2) = 1.$
For \( x = 1 \), the paths have phases \( \phi(1, 0) = 0 \) and \( \phi(1, 1) = \pi \), leading to destructive interference:
$\langle 0 | H H | 1 \rangle = \frac{1}{\sqrt{2}} \cdot \frac{1}{\sqrt{2}} (1 - 1) = \frac{1}{2}(0) = 0.$

For a general circuit, the path amplitude \( \alpha(x, y) \) and the phase polynomial \( \phi(x, y) \) are affected by both Clifford and non-Clifford gates. For instance, a non-Clifford \( T \)-gate introduces a phase shift of \( \frac{\pi}{4} \) when applied to the \( |1\rangle \) state. Specifically, let \( G = \begin{bmatrix} 1 & 0 \\ 0 & e^{i\theta} \end{bmatrix} \) be a non-Clifford gate with angle \( \theta \). When applied to a path variable \( y_i \), the amplitude changes from \( \alpha(x, y) = 1 \) (for Clifford gates) to \( \alpha(x, y) = e^{i\theta y_i} \). For example, the output amplitude for a path passing through a \( T \)-gate is given by $ \sum_{y \in \{0, 1\}^m} e^{i\sum_{i=1}^{m} \frac{\pi}{4} y_i} e^{2\pi i \phi(x, y)}$,
which accumulates phases \( \frac{\pi}{4} \) for each \( y_i = 1 \). This phase accumulation affects the total interference pattern, leading to a modified polynomial \( \phi'(x, y) = \phi(x, y) + \frac{\pi}{4} \sum_{i=1}^{m} y_i \), which changes the constructive and destructive interference conditions. Thus, calculating the overall amplitude involves summing over all paths with modified phase contributions, expressed as: $\left\langle f(x, y) \middle| C \middle| x \right\rangle = \frac{1}{\sqrt{2^m}} \sum_{y \in \{0, 1\}^m} e^{2\pi i \phi'(x, y)}.$
This approach systematically handles arbitrary circuits by tracking how non-Clifford gates affect each path, allowing us to analyze complex interference patterns.
\end{proof}

Clifford+$Rz$ circuits have a phase function $\phi$ that implies a natural equivalence relation: $\phi(x,~y)~ = \phi'(x,~y)~mod~2\pi$. This means that for all $x$ and $y$, the results of $\phi(x,~y)$ and $\phi'(x,~y)$ are equivalent when taking modulo $2\pi$. Sometimes, it is more convenient to represent the phase function using a group isomorphic to $G$ or $G/2\pi\mathbb{Z}$.
For example, when dealing with a broader set of Clifford+$T$ gates, by restricting the phase group to $\pi/4\mathbb{Z}$ might be more convenient. Additionally, due to the isomorphism $\frac{\pi}{4}\mathbb{Z}/2\pi\mathbb{Z} \simeq \mathbb{Z}_{8}$, the phase function of Clifford+$T$ circuits is isomorphic to some $\mathbb{Z}_8$-valued function.

The phase term induced by the $T = R_z^{2\pi/8}$ gate, considering the evolution of subpaths through the $THT$ circuit, as shown in Fig. \ref{fig6}, where $\omega = e^{i\pi/4}$. If we slide the phase of the first T along the subpath, i.e., exchange the first $\omega$ phase with the right Hadamard gate, it duplicates $\omega$ onto the two output paths, effectively acting as a $\omega I$ gate, which adds an additional $\omega$ phase to the paths starting from $\left| 0 \right\rangle$. To maintain the logical functionality of the circuit, one can cancel the extra $\omega$ phase by applying an additional $\omega^{-1}$ phase on $\left| 0 \right\rangle$ (e.g., $X^{\dagger}TX$).

\begin{figure}[h]
\centering
\includegraphics[width=1\linewidth]{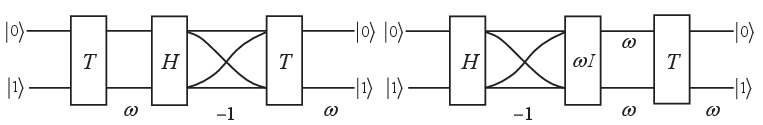}
\caption{Schematic diagram illustrating the paths taken by each input state and the phases acquired along each subpath in $THT$ circuit}
\label{fig6}
\end{figure}

Tracking the evolution of each basis state through the application of phase contributions for each quantum gate, as shown in Fig. \ref{fig8}, the red and blue paths each acquire a total phase of $\pi/2$. This corresponds to the basis states $\left. \left| 01 \right. \right\rangle$ and $\left. \left| 10 \right. \right\rangle$ for which $x_{1} \oplus x_{2} = 1$.

\begin{figure}[h]
\centering
\includegraphics[width=0.9\linewidth]{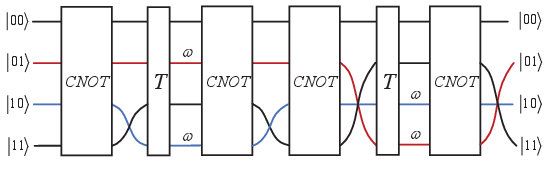}
\caption{Schematic diagram illustrating the paths taken by each input state and the phases acquired along each subpath}
\label{fig8}
\end{figure}
 



\subsection{Incremental Changes of Subpath-Equivalence}
\label{section3.1}
$\Delta$Subpath-equivalence is defined through the subpath sum $SP$. Two quantum circuits are $\Delta$subpath-equivalence if they have two functionally equivalent subpath sums of polynomial size at the same location.

\begin{definition}
\label{def_speSP}
(\(\Delta\)subpath-equivalence based on Subpath Sums \(SP\)): Let \( C_1 \) and \( C_2 \) be two quantum circuits, and let \( SP_1 \) and \( SP_2 \) be their respective subpath sums as defined in Definition 6. The circuits \( C_1 \) and \( C_2 \) are said to be \(\Delta\)subpath-equivalent if there exists a subpath sum \( \Delta SP \subseteq SP \) such that:

\begin{itemize}
\item
   Define the subpath sum operators for the two circuits \( C_1 \) and \( C_2 \) as: 
\begin{equation}
\resizebox{0.43\textwidth}{!}{$
U_{\Delta SP_{1/2}} = \frac{1}{\sqrt{2^m}} \sum_{y \in \mathbb{Z}_2^m} e^{2\pi i \phi_{1/2}(x,y)} \left| f_{1/2}(x,y) \right\rangle \left\langle x \right|$}
\end{equation}

where \( x = x_1 x_2 \cdots x_n \) is the input basis vector (each \( x_i \) is a Boolean constant or variable). \( y = y_1 y_2 \cdots y_m \) are the path variables corresponding to intermediate qubits. \( \phi_1(x, y) \) and \( \phi_2(x, y) \) are the phase polynomials (defined in Definition 6) that describe the phase contribution of the subpath sum for \( C_1 \) and \( C_2 \). \( f_1(x, y) \) and \( f_2(x, y) \) are Boolean polynomials describing the output basis states of the circuits.

\item 
The operators corresponding to the path sums outside \( \Delta SP \) must be identical for both circuits: 
$U_{\Delta SP_{1} \notin SP_1} = U_{\Delta SP_{2} \notin SP_2},$
where \( U_{\Delta SP_{1} \notin SP_1} \) and \( U_{\Delta SP_{2} \notin SP_2} \) are the linear operators defined by the path sums outside the region \( \Delta SP \).

\item 
The subpath sum operators \( U_{\Delta SP_1} \) and \( U_{\Delta SP_2} \) must be equivalent: 
$U_{\Delta SP_1} = U_{\Delta SP_2}.$
\end{itemize}

For Clifford circuits, the phase polynomial \(\phi(x, y)\) is typically a linear or quadratic form, making the operators \(U_{\Delta SP_1}\) and \(U_{\Delta SP_2}\) relatively easy to analyze. However, for general circuits, \(\phi(x, y)\) may include higher-order terms or even non-polynomial forms. This complexity can be handled by subdividing \( \Delta SP \) into smaller segments where each segment \( \Delta SP_k \) corresponds to a simpler, localized phase polynomial. The equivalence condition then translates into ensuring that the total phase evolution along the path is preserved between \( C_1 \) and \( C_2 \). Specifically, if the cumulative phase difference \(\Delta \phi(x, y)\) between \( U_{\Delta SP_1} \) and \( U_{\Delta SP_2} \) satisfies \(\Delta \phi(x, y) \equiv 0 \mod 2\pi\), the two circuits are said to be $\Delta$subpath-equivalent.

\end{definition}

Because we do not require the subset $SP$ to be connected, the path sum can be divided into two different equivalent forms: vertical equivalence and horizontal equivalence. Vertical equivalence occurs when subpath sums are sequentially summed within the same quantum subsystem, maintaining the same cumulative phase contributions. Formally, two circuits $C_1$ and $C_2$ are vertically equivalent if their subpath sum operators can be represented as a sequence of operators $U_{SP_1^1}, U_{SP_1^2}, \ldots, U_{SP_1^k}$ and $U_{SP_2^1}, U_{SP_2^2}, \ldots, U_{SP_2^k}$, respectively, such that for all $j$, $U_{SP_1^j} = U_{SP_2^j}$. The overall equivalence condition is given by $U_{SP_1} = U_{SP_2}$. On the other hand, horizontal equivalence describes equivalent subpath sums located on different subsystems. If there exists a mapping $\sigma: Q_1 \to Q_2$ such that $U_{SP_1} \circ \sigma = U_{SP_2}$, where $Q_1$ and $Q_2$ are the qubit sets on which $SP_1$ and $SP_2$ act, then $C_1$ and $C_2$ are said to be horizontally equivalent.
Since the path sum symbols describe mapping relation between linear combinations of basis vectors, the output $\left. \left| f(x,y) \right. \right\rangle$ of one path sum can be composed with the input $\left. \left| x^{\prime} \right. \right\rangle$ of another path sum by replacing each input value $x_{i}^{\prime}$ with the corresponding output $f_{i}(x,y)$.
For example, for $U_{SP} : \left| x_{1}x_{2}x_{3} \right\rangle \rightarrow \left| x_{1}\left( x_{1}\oplus x_{2} \right)x_{3} \right\rangle$, the calculation of $U_{\Delta SP} : \left| x_{1}^{\prime}x_{2}^{\prime}x_{3}^{\prime} \right\rangle \rightarrow \left| x_{1}^{\prime}x_{2}^{\prime}\left( x_{2}^{\prime}\oplus x_{3}^{\prime} \right) \right\rangle$ can use $x_{1}\bigoplus x_{2}$ instead of $x_{2}^{\prime}$ to combine, as follows $\left| x_{1}x_{2}x_{3} \right\rangle \rightarrow \left| x_{1}\left( x_{1}\oplus x_{2} \right)\left( x_{1}\oplus x_{2}\oplus x_{3} \right) \right\rangle \nonumber, U_{ ~SP + \Delta SP} = U_{SP}U_{\Delta SP}$
This approach allows for the long-distance cancellation of phase gates applied to the same logical state. 

\begin{figure}[h]
  \centering
  \subfloat[Example circuits with 3 qubits]{\includegraphics[width=0.42\linewidth]{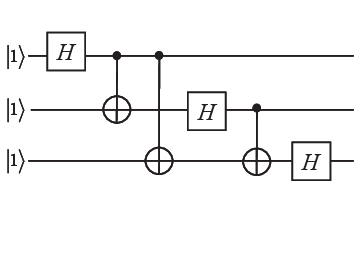}\label{fig14a}}\hspace{0.05\linewidth}
  \subfloat[The corresponding path structure of the example circuit]{\includegraphics[width=0.42\linewidth]{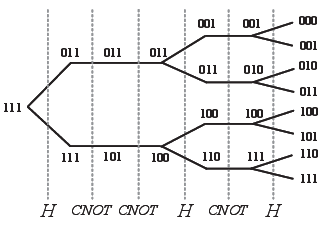}\label{fig15b}}
  \caption{An intuitive representation of semantic transformations. Example circuits with 3-qubits and their corresponding path structures are included.}
   \label{fig14}
\end{figure}

As shown in Fig. \ref{fig14}.(a) and (b), we provide an intuitive representation of a 3-qubits sample quantum circuit in terms of subpath sum semantics. If a quantum gate is branching that can map states to superposition states, such as the $H$ gate, two edges are drawn to represent two different equally weighted subpaths. If the quantum gate is non-branching, a single edge is drawn to represent a single mapping of phase contribution. At the end of each edge, a single state generated by traversing along that path is labeled.

 A loop $\Delta SP$ is represented as a closed path segment in the subpath sum representation, where a set of unitary operations $U_{1}, U_{2}, \ldots, U_{k}$ maps an initial state $\left| \psi_{in} \right\rangle$ back to itself, i.e., $U_{k} U_{k-1} \cdots U_{1} \left| \psi_{in} \right\rangle = \left| \psi_{in} \right\rangle$. In the context of subpath sums, loops introduce phase interference and play a crucial role in controlling the transformation properties of the circuit. By constructing loops strategically, it is possible to achieve phase cancellations or amplifications that are crucial for the indistinguishability properties of QceQiO.

The path structure with added loops and its corresponding quantum circuit are presented as shown in Fig. \ref{fig18}.(a). The loops $\Delta SP$ are added at different positions in the path structure, yet they still form loops between the blue and red lines, achieving long-distance phase cancellation. The added $\Delta SP$ can also be represented in the quantum circuit as quantum gates. In addition, as shown in Fig. \ref{fig18}.(b), parts of the added loops can also act on the corresponding qubit to change quantum state of quantum implementation. 
We all take half of the loop and act it to the quantum state, so that on the one hand, it is convenient to express the half of the loop with $U^{\dagger}$ and $U$ when writing, and on the other hand, the quantum gate formed in this way has low construction cost under the existing quantum system. In fact, each part can be divided unequally.

\begin{figure}[h]
  \centering
    \subfloat[Adding loops to path structure for long-distance phase cancellation]{\includegraphics[width=0.42\linewidth]{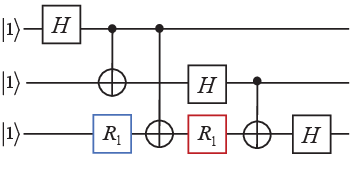}\label{fig:a}\hspace{0.05\linewidth}
  \includegraphics[width=0.43\linewidth]{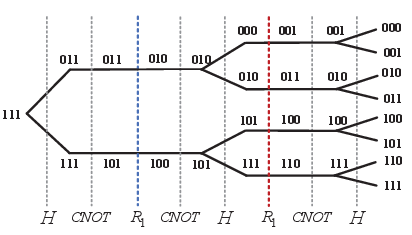}\label{fig:b}}
  \newline
   \subfloat[Incorporating loops into quantum states from path structure]{\includegraphics[width=0.42\linewidth]{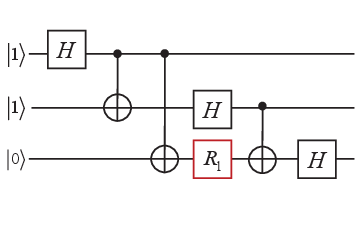}\label{fig:a}\hspace{0.05\linewidth}
  \includegraphics[width=0.43\linewidth]{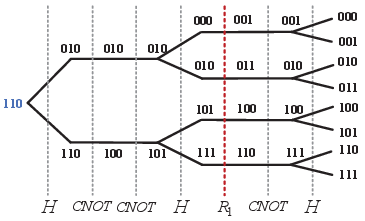}\label{fig:b}}
  \caption{Path structure with added loops and its corresponding example circuit.}
   \label{fig18}
\end{figure}

Then, we propose a specific construction scheme for QceQiO. The QceQiO obfuscator selects the identity subpaths sum that can form a loop as the incremental unit to construct the quantum logic obfuscation circuit for the corresponding target.

\begin{figure}[h]
\centering
\includegraphics[width=0.95\linewidth]{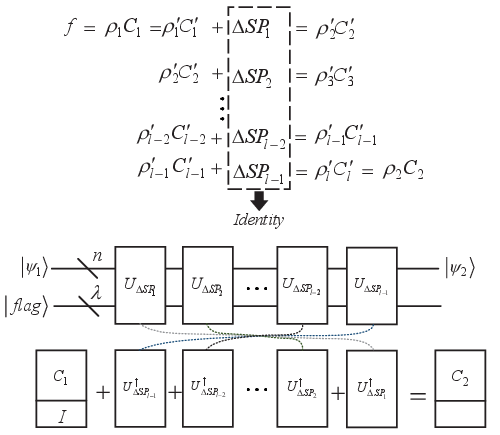}
\caption{Construction of QceQiO}
\label{fig10}
\end{figure}

\begin{definition}
\label{consqceQiO}
(Construction of QceQiO): Two quantum circuit clusters $\{C_{\lambda}^{1}\}_{\lambda \in \mathbb{N}}$ and $\{C_{\lambda}^{2}\}_{\lambda \in \mathbb{N}}$ are considered obfuscatable if there exists a polynomial $\ell = \ell(n)$ and a function $B = O(logn)$ such that, for any positive integer $n$, there are $\ell \geq 2$ intermediate circuits $C_{1}^{\prime},C_{2}^{\prime},\cdots,C_{\ell}^{\prime}$ and a series of subpaths $\left\{ {\mathrm{\Delta}SP}_{i} \right\}_{i \in {\lbrack{\ell - 1}\rbrack}}$ of size $B$. Here, $C_{1}^{\prime} = C_{\lambda}^{1},C_{\ell}^{\prime} = C_{\lambda}^{2}$, and for any $i \in \left\lbrack \ell - 1 \right\rbrack$, $C_i$ is $\delta$-equivalent to $C_{i+1}$ through the subpath ${\mathrm{\Delta}SP}_{i}$.
\end{definition}

In simpler terms, two quantum circuits are considered logically obfuscatable if one can be transformed into the other through a polynomial number of steps. Each step involves modifying a subpath of polynomial size while guaranteeing the functional equivalence of subpath sums.
The specific construction of QceQiO is shown in Fig. \ref{fig10}. The QceQiO takes as input the quantum implementation $(\rho,C)$ of a function $\left. f:x\rightarrow y \right.$ and the parameter $1^{\lambda}$, where $C$ is a universal quantum circuit.
\begin{itemize}
\item{Step 1.} Prepare a new quantum state $\rho \otimes \left( \left. \left| flag \right. \right\rangle^{\otimes \lambda} \right)$ and a quantum circuit $C \otimes \left( I^{\otimes \lambda} \right)$ based on the quantum state $\rho$. The flag, denoted as $\left| flag \right\rangle^{\otimes \lambda}$, is an auxiliary quantum state used to control and track subpath sum transformations in the QceQiO construction. It serves as a marker to ensure that each subpath transformation is applied coherently and consistently. The cost of preparing and manipulating the flag is $O(\log n)$ single-qubit gates, while the cost of applying $U_{\Delta SP_i/2}$ to the flag and input state is polynomial in the number of qubits.

\item{Step 2.} For each randomly chosen $\mathrm{\Delta}{SP}_{i}$, $i \in \left\{ {1,\ldots,\ell} \right\}$, where each $\mathrm{\Delta}{SP}_{i}$ is a subpath sum of the identity transformation $I$. Apply half of it to $\rho \otimes \left( \left. \left| flag \right. \right\rangle^{\otimes \lambda} \right)$, resulting in the quantum state $\rho_{i}^{\prime} = U_{\mathrm{\Delta}{SP}_{i}/2}\left( \rho \otimes \left( \left. \left| flag \right. \right\rangle^{\otimes \lambda} \right) \right)$. As in Definition \ref{def_speSP}, it is not required that each $\Delta SP_i$ is connected, allowing for long-distance cancellation of phase gates applied to the same logical state. Each loop can be inserted at any intermediate point in the quantum subpath structure. 
\item{Step 3.} Insert the other half of the subpath $\mathrm{\Delta}{SP}_{i}$ into the paths of $C \otimes I^{\otimes \lambda}$, forming a new quantum circuit $C_{i}^{\prime} = U_{\mathrm{\Delta}{SP}_{i}/2}\left( C \otimes \left( I^{\otimes \lambda} \right) \right)$.
Similarly, other single-qubit gates related to phase, such as Z and S gates, can be constructed in this manner by adding identity transformations in the free slots of the quantum circuit, where the corresponding qubits should be in an idle/waiting state.
\item{Step 4.} Repeat Steps 2 and 3 $\ell$ times.
\item{Step 5.} Output the obfuscated quantum implementation $\left( \rho^{\prime},C^{\prime} \right)$. 
\end{itemize}

According to Definitions \ref{def_speSP}, \ref{qceQiO}, \ref{consqceQiO}, each intermediate implementation $\left( \rho_{i}^{\prime},C_{i}^{\prime} \right)$ is constructed through $\mathrm{\Delta}{SP}_{i}$ to achieve $\Delta$subpath-equivalence. This ensures the quantum implementations $\left( {\rho,C} \right)$ and $\left( \rho^{\prime},C^{\prime} \right)$ are obfuscatable.

\subsection{Distinguisher Reduction to Probability Test}
\label{section4.2}

To mitigate the security loss mentioned in the Definition \ref{qceQiO}, the functional equivalence of QiO can be probability test through a polynomial identity test on the subpath sums. In essence, random sampling of path variables of the quantum circuit is performed, and the subpath sums are checked for equality. If they are not equal, we have a counterexample. if they are equal, there is a high probability that the two quantum circuits are equivalent. We take inspiration from the idea of polynomial identity checking, which has applications in integrated circuit verification and semantics optimization. The correctness of polynomial identity testing is derived directly from the Schwartz-Zippel lemma \cite{Motwani_Raghavan_1995} and has been extended in the context of semantics optimization \cite{2023Synthesizing}. Here we reduce indistinguishability verification of quantum circuits problem into subpath sums equivalence problem and apply a probability test on the phase polynomials of subpath sums.
\begin{theorem}
\label{theo_protest}
(Probability Test for Subpath Sums): Consider two phase polynomials $\phi_{1}$ and $\phi_{2}:\phi \in D_{M}\left\lbrack {x,y} \right\rbrack$, involving the same set of $n$ variables. Let insert $\ell$ polynomials into $\phi_{1}$. Check whether $\phi_{1}(v) = \phi_{2}(v)$ holds for all variables $v \in \mathbb{C}^{n}$, succinctly denoted by $\phi_{1} = \phi_{2}$. Use $d$ to represent the maximum degree of these two polynomials.
Thus we can probabilistically verify the equivalence of $\phi_{1}$ and $\phi_{2}$ with high probability $\left. 1 - \ell^{2}d/ \middle| R \right|$.
\end{theorem}

\begin{proof}
Suppose we insert $\ell$ polynomials into $\phi_{1}$ that are inequivalent to each other on the same set of $n$ variables. Implicitly, the following random computation is performed:
\begin{itemize}
\item {Step 1.} Choose a finite subset $R \subset \mathbb{C}$.
\item {Step 2.} Draw $n$ independent values uniformly from the distribution over $R$, denoted by $v_{1},~.~.~.~,~v_{n}$.
\item {Step 3.} For each pair $\phi_{i}$ and $\phi_{j}$, where $i~ \neq ~j$, $i,j \in \left\{ 1,\ldots,\ell \right\}$, check whether $\phi_{i}(v) = \phi_{j}(v)$.
\end{itemize}

By the Schwartz-Zippel lemma (worst-case guarantee), if $\phi_{1} = \phi_{2}$, the algorithm returns True. This is because for any $v$, $\phi_{1}(v) = \phi_{2}(v)$. If $\phi_{1} \neq \phi_{2}$, then it returns True with a probability of at most $\left. d/ \middle| R \right|$, i.e., the probability that $\phi_{1}(v) = \phi_{2}(v)$ is at most $\left. d/ \middle| R \right|$. Conversely, it returns False with a probability of at least $\left. 1-d/ \middle| R \right|$. For $\phi_{1} \neq \phi_{2}$, the algorithm might return an incorrect answer (True) with low probability.

Thus, according to the union bound, for any pair of polynomials $\phi_{i}$ and $\phi_{j}$, where $i~ \neq ~j$, $i,j \in \left\{ 1,\ldots,\ell \right\}$. if $\phi_{i}(v) = \phi_{j}(v)$, return True. if $\phi_{i}(v) \neq \phi_{j}(v)$, the probability of returning True is  
\begin{equation}
\scalebox{1}{$
\begin{split}
& \Pr\left[\exists i \neq j. \phi_{i}(v) = \phi_{j}(v)\right] \\
&\leq \sum_{i,j \in [1,\ell],i \neq j} \Pr\left[\phi_{i}(v) = \phi_{j}(v)\right]
\leq {\ell^{2}d}/{|R|}
\end{split}$}
\end{equation}
Conversely, it returns False with a probability of at least $\left. 1 - \ell^{2}d/ \middle| R \right|$.
\end{proof}
For example, suppose $\phi_{1}$ and $\phi_{2}$ are inequivalent polynomials of degree 10. If we choose $R$ to be the set of 64 bits integers $({|R|}=2^{64} \approx 10^{20})$, we would have a failure probability of approximately $10^{-19}$. If we insert $10^{6}$ polynomials, each with a maximum degree of 10 and mutually inequivalent, into a new phase polynomial and use 64 bits integers as $R$, the probability of declaring a pair of polynomials as equivalent is approximately $10^{-7}$. It can be observed that if the polynomials are not equivalent, the total failure probability increases due to the possibility of false positives in each verification. It still provides a relatively low failure probability, considering factors such as the actual decoherence capabilities in quantum computing. The actual number of inserted phase polynomials for subpath sums is much smaller than $10^{6}$.

\section{Security analysis}
\label{section5}
In this section, we present the security analysis of QceQiO. The primary tools for the proof include the `Clifford twirl' \cite{ABOEM17} and the `admissible oracle lemma' \cite{GJMZ23}. The `Clifford twirl' is sufficient to establish security against adversaries performing a single query, but a more intricate argument is required to handle general adversaries. This argument leverages the `admissible oracle lemma'. We prove that the two quantum implementations of the classical function $f$ using QceQiO are indistinguishable in the quantum random oracle model. This result aligns with previous research in the field, which has extensively examined the security of quantum-cryptographic constructions in random oracle model \cite{Zhang22,AC12,BM22}.
\subsection{Scenario and Adversary-Challenger Game}
In this section, we consider a classical scenario in quantum cloud delegation computation \cite{Verification19} to illustrate the practical setting for QceQiO’s obfuscation. Quantum delegation protocols typically involve two parties: a verifier and a prover. The verifier’s goal is to delegate the execution of a complex quantum circuit $C$ to a prover while ensuring that the computation is performed correctly and securely. The verifier begins by preparing and sending the quantum state $| \phi \rangle$ and $C$ to the prover, who performs a unitary operation corresponding to $C$.

In the context of QceQiO, the delegation scenario serves as a security framework: the verifier wants to obfuscate the circuit $C$ in such a way that the prover cannot gain any knowledge about its structure or functionality, even if they have access to the obfuscated circuit. The verifier will then interact with the prover to confirm, with high probability, that the correct (obfuscated) circuit has been applied to the input. If any discrepancy is detected, the protocol will be aborted.

The concept of quantum copy protection based on hidden subspace oracles was introduced by Aaronson et al. \cite{AC12} and was later extended to the black box model \cite{ALL+21}. We adapt it here to suit quantum indistinguishability obfuscation.

\begin{definition}
\label{definition9}
(Quantum Random Oracle for Quantum Copy Protection): A quantum random oracle is a quantum polynomial-time algorithm  that satisfies following properties:

\begin{itemize}
\item Oracle Key Selection: The challenger randomly selects a key \( k \) and shares \( k \) with the oracle \( R \). Each possible key \( k \) specifies a unique \( n \)-qubit Haar-random unitary operation, denoted by \( U_n^k \), $k \rightarrow U_n^k$. The key length \( |k| \) is determined as \( O\left( \log\left( \left| U(2^n) \right| \right) \right) \).

\item Query Procedure: Assume there exists a quantum implementation \( \left( \left| \psi_{0} \right\rangle, C_{f} \right) \) of a function \( f \). An adversary queries the oracle \( R_{1},...,R_{q} \) with \( q \)-quantum queries. Each query begins with an initial state \( \left| \psi_{0} \right\rangle \), and the adversary applies a sequence of unitary transformations \( U_{0},...,U_{q} \) successively between the queries. Each \( U_i \) is a Haar-random unitary \( U_n^{k_i} \). The transformations are successively applied as $\left. U_{q}R_{q}~...U_{1}R_{1}U_{0} \middle| \psi_{0} \right\rangle.$

\item Oracle Application and Reverse Operations: In each query \( i \), the oracle \( R \) applies \( U_{i}^{\dagger}C_{f} \) to the received quantum state \( \left. U_{i} \middle| \psi_{0} \right\rangle \). The oracle undoes the effect of the Haar-random unitary \( U_n^{k} \) between queries, recovering the quantum implementation of function $f$.
\end{itemize}
\end{definition}

In the context of the quantum random oracle model described above, we design the adversary-challenger game for quantum copy protection.

\begin{definition}
(Quantum Copy Protection Adversary-Challenger Game):
\begin{itemize}
\item The challenger randomly selects a function $f$ from a set of classical circuits $\mathcal{F} = \left\{ \mathcal{F}_{\lambda} \right\}_{\lambda \in \mathbb{N}}$, where $\lambda$ is the security parameter. The challenger creates the quantum implementation $\sigma = CP(1^\lambda,f)$ of function $f$ using the copy protection scheme $CP$.
\item The challenger sends the quantum implementation $\sigma$ to the adversary, who maps this quantum implementation to a quantum state $\rho_{AB}$ on two registers $\mathcal{A}$ and $\mathcal{B}$ corresponding to two different subsystems.
The adversary sends to the challenger a description of the quantum circuit families for the subsystems $\mathcal{A}$ and $\mathcal{B}$, denoted by $C_A$ and $C_B$. 
\item The challenger executes the verification circuits ${Ver}_{\lambda}\hspace{0pt}(f,C_{A}) \otimes {Ver}_{\lambda}(f,C_{B})$ on the quantum state $\rho_{AB}$, where ${Ver}_{\lambda}$ is a polynomial-sized quantum circuit taking $f$ as input, along with the quantum implementation of $f$. It outputs a single bit. If both outputs of the verification circuits are 1, the challenger outputs 1, otherwise, it outputs 0.
\end{itemize}
\end{definition}
Copy protection is said to be $(Ver,\epsilon)$-secure if, for all quantum polynomial-time algorithms $Adv$, there exists a negligible function $negl$, such that for all $~\lambda$,
\begin{equation}
\Pr\left\lbrack Adv(CP\hspace{0pt}(1^\lambda,f) = \sigma )= 1 \right\rbrack \leq \epsilon(\lambda) + negl(\lambda).
\end{equation}

This game is designed to assess whether the copy protection scheme $CP$ can resist attacks from the adversary (Adv), where the verification circuit $Ver$ is employed to verify the correctness of the quantum implementation. The security of the game is controlled by the parameter $\epsilon$, representing the permissible error probability allowed by the challenger during the verification process.

Consider a quantum channel $Enc_s$(encoding) acting on the verifier's state, where s denotes a private random string sampled by the verifier from some distribution $p(s)$. Let Phonest represent the $CP$-complete mapping of the verifier's state, considering the honest behavior of the prover following the verifier's instructions in the protocol. Additionally, define ${P_{\hspace{0pt}}}_{incorrect}^{s} = \left( I - \middle| \Psi_{out}^{s} \right\rangle\left\langle \Psi_{out}^{s} \middle| \right)\left. \otimes \middle| {acc}^{s} \right\rangle\left\langle {acc}^{s} \right|.$
as the projection onto the orthogonal complement of the correct output \resizebox{0.5\textwidth}{!}{$
\left| \Psi_{out}^{s} \right\rangle\left\langle \Psi_{out}^{s} \middle| = {Tr}_{flag}\hspace{0pt}\left( \mathcal{P}_{honest}\left( {Enc}_{s}\left( \middle| \psi \right\rangle\left\langle \psi \middle| \otimes \middle| acc \right\rangle\left\langle acc \middle| \right) \right) \right) \right.$}
and, in the case of an accepted flag state \resizebox{0.5\textwidth}{!}{$
\left| {acc}^{s} \right\rangle\left\langle {acc}^{s} \middle| = {Tr}_{input}\hspace{0pt}\left( \mathcal{P}_{honest}\left( {Enc}_{s}\left( \middle| \psi \right\rangle\left\langle \psi \middle| \otimes \middle| acc \right\rangle\left\langle acc \middle| \right) \right) \right) \right.$}

Then, we say that such a protocol is $\epsilon$-Security: Consider a quantum channel $Enc_s$(encoding) acting on the verifier's state, where s denotes a private random string samp-verifiable (where $0 \leq \beta \leq 1$), if for any action $\mathcal{P}$ by the prover, we have, \resizebox{0.5\textwidth}{!}{$
Tr\left( {\sum\limits_{s}{p(s){P_{\hspace{0pt}}}_{incorrect}^{s}\mathcal{P}\left( {Enc}_{s} \middle| \psi \right\rangle\left\langle \psi \middle| \otimes \middle| acc \right\rangle\left\langle acc \middle| \right)}} \right)) \leq \epsilon.$}

Essentially, this definition indicates that the probability for the output of protocol to be incorrect and the verifier accepting, should be bounded by $\epsilon$.

\subsection{Security Analysis of Scheme}
\begin{theorem}
\label{theo_sec}
The two quantum implementations of a classical function $f$ under QceQiO are indistinguishable in quantum random oracle model.
\end{theorem}
\begin{proof}
The proof is achieved through a series of hybrid arguments, where we demonstrate that Adv's success probability remains constant across these hybrids. Hybrid 0 corresponds to the original security game of copy protection. Hybrid 1 corresponds to a pure version of the game, and the adversary is allowed only a single query. Hybrid 3 is identical to Hybrid 1, except that the adversary can access the oracle for an arbitrary number of queries. Hybrid 2 and 4 replaces the Clifford unitaries in Hybrid 1 and 3 with equivalent subpaths.  Similarly, Hybrid 4 is the same as Hybrid 2, with the additional allowance for the adversary to make an arbitrary number of queries to the oracle. 

Hybrid 0: Adversary-challenger game for quantum copy protection.

Hybrid 1: Purified version of the adversary-challenger game for quantum copy protection.

Let $\left( \left. \left| \psi_{0} \right. \right\rangle,U_{0} \right)$ and $\left( \left. \left| \psi_{1} \right. \right\rangle,U_{1} \right)$ be two quantum implementations of the same classical function $f$. The challenger randomly selects a bit $b \in \left\{ 0,~1 \right\}$, determining the corresponding quantum implementation based on the chosen value of $b$. Then, the challenger creates a quantum state $\left| \Psi_{in} \right\rangle\left. ~ = \middle| \psi_{b} \right\rangle\left. \otimes \middle| flag \right\rangle$, where $\left| \psi_{b} \right\rangle$ is the input state of $n$ qubits, and $\left| flag \right\rangle$ is a $\lambda$-qubit flag state representing whether the challenger accepts $(\left| flag \right\rangle\left. ~ = ~ \middle| acc \right\rangle)$ or rejects $\left| flag \right\rangle\left. ~ = ~ \middle| rej \right\rangle$ at the end. Let $t = n + \lambda$, and $C_t$ is the set of $t$-qubit unitary operations. Before the oracle query, the challenger applies the average of all possible unitary units $C$ to $\left| \Psi_{in} \right\rangle$. Specifically, the challenger creates the following quantum state, $\left| \Psi_{1} \right\rangle~ = \frac{1}{\left| C_{t} \right|}{\sum\limits_{C \in \mathcal{C}}C}\left( \left. \left| \psi_{b} \right. \right\rangle \otimes \middle| flag \right\rangle).$

The challenger hands over the created quantum state to adversary $A$, who can query the oracle $R$. Adversary $A$ attempts to guess the value of $b$ randomly chosen by the challenger by observing the output of oracle $R$ and the operations on the quantum register. Adversary $A$ wins if their guessed value $b^{\prime}$ matches the $b$ randomly chosen by the challenger in the first step.
Note that applying the average of all possible unitary units $C$ is equivalent to applying a randomly chosen unitary unit. This is because the set of unitary units is closed under average operations. In other words, for any two unitary units $C_1$ and $C_2$, their average is still a unitary unit. Therefore, the average operation over all possible unitary units can be considered as applying a randomly chosen unitary unit.
\begin{lemma}
For any adversary $A$, the adversary has equal advantage in Hybrid 0 and 1, ${\Pr\left\lbrack {A~wins~in~Hybrid~0} \right\rbrack} = {\Pr\left\lbrack {A~wins~in~Hybrid~1} \right\rbrack}.$
\begin{proof}
This is evident since Hybrid 1 is a purification of Hybrid 0.
\end{proof}
\end{lemma}
\begin{lemma}
For any adversary $A$, there exists a negligible function $negl$, such that for all $\lambda$, $\Pr\lbrack A~wins~in~Hybrid~1\rbrack \leq \epsilon(\lambda) + negl(\lambda).$
\end{lemma}
\begin{proof}
For any completely positive and trace-preserving (CPTP) map, it can be represented through a Kraus decomposition, $\mathcal{E}(\rho) = {\sum\limits_{i}K_{i}}\hspace{0pt}\rho K_{i}^{\dagger},$
 where $K_i$ is a set of Kraus operators satisfying ${\sum\limits_{i}{K_{i}^{\dagger}K}_{i}}\hspace{0pt} = I\hspace{0pt}$, and $\rho$ is an $n$-qubit density matrix. The map $\mathcal{E}$ preserves the trace of quantum states, i.e., $Tr\left\lbrack \mathcal{E}(\rho) \right\rbrack = Tr\lbrack\rho\rbrack$, for all quantum states $\rho$, ensuring that the map does not introduce or deplete quantum probability.
The $n$-qubit Pauli group $\mathbb{P}_{n}$ forms a basis for all $2^n\times 2^n$ matrices, and every Kraus operator can be represented as $K_{i} = {\sum\limits_{j}{\alpha_{ij}P_{j}}},$
where $j$ ranges over all indices of $n$-qubit Pauli operators, and $\alpha_{ij}$ satisfies ${\sum\limits_{j}{\alpha_{ij}{\hspace{0pt}\alpha}_{ij}^{*}}} = 1$. Therefore, $\mathcal{E}(\rho) = {\sum\limits_{ijk}{\alpha_{ij}{\hspace{0pt}\alpha}_{ik}^{*}P_{j}\hspace{0pt}\rho P_{k}}}.$
It can be observed that any general map can be expressed as a combination of Pauli operators acting on both sides of the target state $\rho$. Importantly, the Pauli operators on both sides do not need to be identical.
\begin{lemma}
\label{lem_CTD}
(Clifford Twirl Distortion \cite{ABOEM17}): Let $P_1$ and $P_2$ be two operators from the $n$-qubit Pauli group satisfying $P_{1} \neq P_{2}$. For any $n$-qubit density matrix $\rho$, the following relation holds $\sum\limits_{C \in C_{n}}{C^{\dagger}P_{1}\hspace{0pt}C\rho C^{\dagger}P_{2}\hspace{0pt}C = 0}.$

\end{lemma}
Before the oracle query, the challenger averages the effect of all possible Clifford operators $C$ on the state $\left| \Psi_{in} \right\rangle$. In this case, the state received by the adversary is given by $\rho = \frac{1}{\left| C_{t} \right|}{\sum\limits_{l}{\left. C_{l} \middle| \Psi_{in} \right\rangle\left\langle \left. \Psi_{in} \right| \right.C_{l}^{\dagger}}}.$
The adversary performs a dishonest operation $\mathcal{E}$ on this state and uses the above formula, yielding $\frac{1}{\left| C_{t} \right|}{\sum\limits_{ijkl}{\left. \alpha_{ij}{\hspace{0pt}\alpha}_{ik}^{*}P_{j}C_{l} \middle| \Psi_{in} \right\rangle\left\langle \left. \Psi_{in} \right| \right.C_{l}^{\dagger}P_{k}}}.$
Subsequently, the adversary takes this state and applies the oracle operation, which involves the reversal of the Clifford applied by the challenger, $\frac{1}{\left| C_{t} \right|}\sum_{ijkl}\left(\alpha_{ij}\alpha_{ik}^{*}C_{l}^{\dagger}P_{j}C_{l} | \Psi_{in}\rangle\langle\Psi_{in}|C_{l}^{\dagger}P_{k}C_{l} \right).$ 

According to Lemma \ref{lem_CTD}, all terms where different Pauli operations act on the two sides (i.e., $j \neq k$) will vanish. The result is $\left| \Psi_{out} \right\rangle\left\langle \Psi_{out} \right| = \frac{1}{\left| C_{t} \right|}\sum_{ijl}\left(\alpha_{ij}\alpha_{ij}^{*}C_{l}^{\dagger}P_{j}C_{l}\left| \Psi_{in} \right\rangle\left\langle \Psi_{in} \right|C_{l}^{\dagger}P_{j}C_{l}\right).$

Each Pauli term is conjugated by a Clifford operator from the set $C_t$, implying that for each Pauli term $P$, there exists a Clifford operator $C$ such that $C_{l}^{\dagger}P_{j}C_{l}$ is another Pauli term. That is, ${\sum\limits_{l}{C_{l}^{\dagger}P_{j}C}_{l}} = {\sum\limits_{P \in \mathbb{P}_{t}}P}$, and since ${\sum_{ij}{\alpha_{ij}{\hspace{0pt}\alpha}_{ij}^{*}}} = 1$, we obtain $
\left| \Psi_{out} \right\rangle\left\langle \left. \Psi_{out} \right| \right. = \beta\left| \Psi_{in} \right\rangle\left\langle \left. \Psi_{in} \right| \right. + \frac{1 - \beta}{4^{t} - 1}{\sum\limits_{i,P_{i} \neq I}{P_{i}\left| \Psi_{in} \right\rangle\left\langle \Psi_{in} \middle| P_{i} \right.}},$
where $\beta$ is the coefficient that yields the ideal outcome, $0 \leq \beta \leq 1$. On each qubit $t=n+\lambda$, there are four possible Pauli operations, $\{I,X,Y,Z\}$, excluding the case of $I_t$. Finally, computing 
$Tr\left( P_{incorrect}\hspace{0pt} \middle| \Psi_{out} \right\rangle\left. \left\langle \left. \Psi_{out} \right| \right. \right)= \frac{1 - \beta}{4^{t} - 1}{\sum\limits_{i,P_{i} \neq I}{Tr\left( P_{incorrect}\hspace{0pt}P_{i}\left| \Psi_{in} \right\rangle\left\langle \Psi_{in} \middle| P_{i} \right. \right)}}$.

This term is non-zero only when $P_i$ acts as the identity on the flag subsystem $\left| flag \right\rangle$ in $\left| \Psi_{in} \right\rangle$. The number of such terms can be accurately calculated as $4^{n}2^{\lambda} - 1$, which yields $
Tr\left( P_{incorrect}\hspace{0pt} \middle| \Psi_{out} \right\rangle\left. \left\langle \left. \Psi_{out} \right| \right. \right) \leq \left( {1 - \beta} \right)\frac{4^{n}2^{\lambda} - 1}{4^{n + \lambda\hspace{0pt}}} \leq \frac{1}{2^{\lambda}}\hspace{0pt}.$
\end{proof}

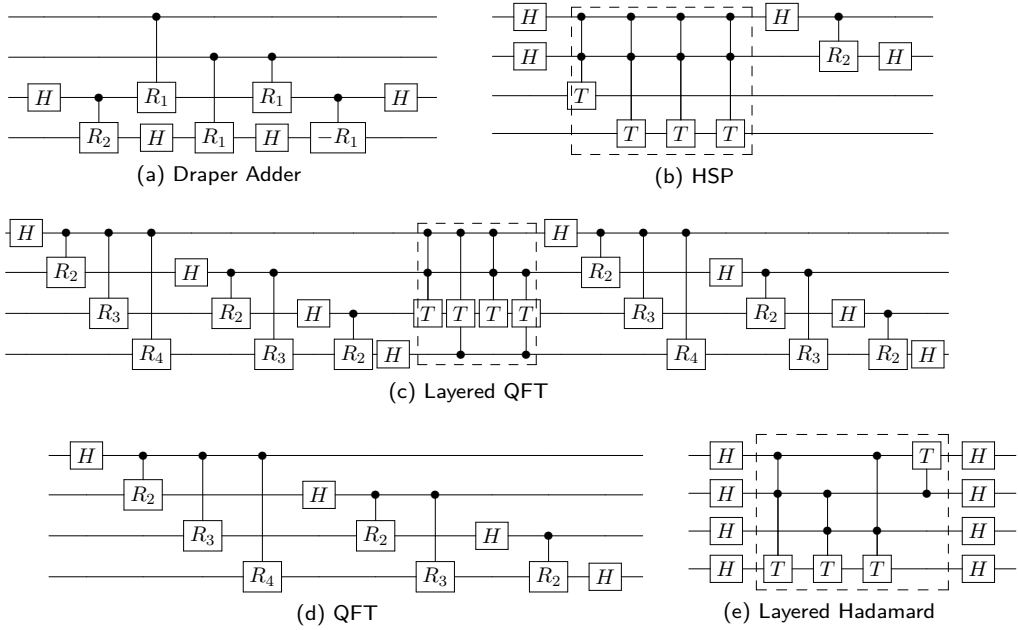
\begin{figure*}[htp]
\centering
\subfloat[Draper Adder]{
\scalebox{0.73}{
\Qcircuit @C=0.91em @R=0.5em @!R=1.4em {
& \qw & \qw        & \ctrl{2} & \qw  & \qw & \qw& \qw& \qw\\
& \qw & \qw        & \qw & \ctrl{2}& \ctrl{1} & \qw& \qw& \qw\\
& \gate{H} & \ctrl{1} & \gate{R_1} & \qw & \gate{R_1}& \ctrl{1}& \gate{H}& \qw\\
& \qw & \gate{R_2} & \gate{H} &\gate{R_1} & \gate{H}& \gate{-R_1}& \qw& \qw\\
}}
}
\subfloat[HSP]{
\quad
\scalebox{0.73}{
\Qcircuit @C=1em @R=0.5em {
& \gate{H} & \ctrl{2}  & \ctrl{3} & \ctrl{3}  & \ctrl{3} & \gate{H}& \ctrl{1}& \qw& \qw\\
& \gate{H} & \ctrl{1}  & \ctrl{2} & \ctrl{2}& \ctrl{2} & \qw&\gate{R_2}& \gate{H}& \qw\\
& \qw & \gate{T}    & \qw        & \qw      & \qw    & \qw & \qw& \qw& \qw\\
& \qw & \qw         &  \gate{T} &\gate{T} & \gate{T} & \qw& \qw& \qw& \qw \gategroup{1}{3}{4}{6}{.7em}{--}\\
}}
}
\newline
\centering
\subfloat[Layered QFT]{
\scalebox{0.73}{
\Qcircuit @C=0.21em @R=0.5em {
& \gate{H} & \ctrl{1} & \ctrl{2} & \ctrl{3} & \qw &\qw & \qw & \qw & \qw & \qw & \ctrl{2} & \ctrl{2} & \ctrl{2} &\qw  & \gate{H} & \ctrl{1} & \ctrl{2} & \ctrl{3} & \qw &\qw & \qw & \qw & \qw & \qw & \qw\\
& \qw   &  \gate{R_2} & \qw & \qw &\gate{H} & \ctrl{1} & \ctrl{2} & \qw & \qw & \qw& \ctrl{1} & \qw & \ctrl{1}& \ctrl{1} &\qw  & \gate{R_2} & \qw & \qw &  \gate{H}& \ctrl{1} &  \ctrl{2} & \qw & \qw & \qw & \qw\\
& \qw & \qw & \gate{R_3} & \qw & \qw &\gate{R_2} & \qw & \gate{H} & \ctrl{1} & \qw &  \gate{T} & \gate{T} & \gate{T} & \gate{T}& \qw &  \qw & \gate{R_3} & \qw & \qw & \gate{R_2}& \qw & \gate{H} & \ctrl{1} & \qw & \qw\\
& \qw & \qw & \qw & \gate{R_4} & \qw &\qw & \gate{R_3} & \qw & \gate{R_2} & \gate{H} & \qw &\ctrl{-1} &  \qw & \ctrl{-1} & \qw & \qw & \qw & \gate{R_4} & \qw & \qw & \gate{R_3}  & \qw & \gate{R_2} & \gate{H} & \qw \gategroup{1}{12}{4}{15}{.7em}{--}\\
}}}
\newline
\centering
\subfloat[QFT]{
\scalebox{0.73}{
\Qcircuit @C=1em @R=0.5em {
& \gate{H} & \ctrl{1} & \ctrl{2} & \ctrl{3} & \qw & \qw& \qw& \qw& \qw& \qw& \qw\\
& \qw & \gate{R_2}   & \qw       & \qw & \gate{H} & \ctrl{1}& \ctrl{2}& \qw& \qw& \qw& \qw\\
& \qw & \qw          & \gate{R_3} & \qw &  \qw   & \gate{R_2}& \qw& \gate{H}& \ctrl{1}& \qw& \qw\\
& \qw & \qw         & \qw       &\gate{R_4} & \qw & \qw    & \gate{R_3}& \qw& \gate{R_2}& \gate{H}& \qw\\
}}}
\subfloat[Layered Hadamard]{
\quad
\scalebox{0.73}{
\Qcircuit @C=1em @R=0.5em {
& \gate{H} & \ctrl{3}  & \qw  & \ctrl{3}  &\gate{T} & \gate{H}& \qw\\
& \gate{H} & \ctrl{2}  & \ctrl{2} & \qw& \ctrl{-1} & \gate{H}& \qw\\
& \gate{H} & \qw    & \ctrl{1}   & \ctrl{1}      & \qw    & \gate{H} & \qw\\
& \gate{H} & \gate{T}  &  \gate{T} &\gate{T} & \qw& \gate{H}& \qw \gategroup{1}{3}{4}{6}{.7em}{--}\\
}}}

\caption{Circuits for Benchmarks one. Dashed boxes in figure represent randomly generated sets of Toffoli gates.}
\label{fig11}
\end{figure*}

Hybrid 2: The challenger applies a randomly chosen unitary to $\left| \Psi_{in} \right\rangle$ before the oracle query, replacing it with an equivalent subpath, while the rest remains the same as Hybrid 1.
\begin{lemma}
\label{lem4}
For any adversary $A$ in Hybrid 2, there exists a negligible function $negl$, such that for all $\lambda$, $\left| \Pr\left\lbrack  {A~wins~in~Hybrid~2} \right\rbrack \right. 
 - \left. \Pr\left\lbrack {A~wins~in~Hybrid~1} \right\rbrack \right| \leq \text{negl}(\lambda).$
\end{lemma}

\begin{proof}
\label{lemma_4}
In Hybrid 2, the Haar-random unitary operations used in Hybrid 1 are replaced by subpaths sampled according to a refined path sum distribution $\mathcal{D}_{\Delta SP}$. Unlike standard path sums, here we focus on paths whose net effect is equivalent to the identity operator. This distinction arises because, in Hybrid 1, each Haar-random unitary $U$ applied to the quantum state $\rho$ is immediately followed by its inverse $U^{\dagger}$ when interacting with the classical implementation $C_f$. Consequently, the effective action on the overall circuit is nullified, and we need only consider equivalent paths whose cumulative phase and state evolution correspond to the identity operation. Thus, it suffices to replace each $U$ and $U^{\dagger}$ pair with a subpath $\Delta SP$ that has a zero net effect.

More formally, for a given unitary transformation $U$, we define a subpath sum operator $U_{\Delta SP}$ that captures the incremental effect of identity gates along the selected path segment. Each $\Delta SP$ is constructed to preserve the input-output relations of the original $U$-$U^{\dagger}$ pairs. The distribution $\mathcal{D}_{\Delta SP}$ is defined over the set of paths $\Delta SP$ such that $U_{\Delta SP}$ is functionally equivalent to the identity. This requirement ensures that the replacement of $U$ and $U^{\dagger}$ with $\Delta SP$ does not alter the adversary’s ability to distinguish Hybrid 1 from Hybrid 2.
Since the identity-equivalent paths are closed under unitary conjugation, the distribution $\mathcal{D}_{\Delta SP}$ guarantees uniform coverage over all such equivalent paths, preserving the input-output behavior of the quantum implementation.

To analyze the adversary's advantage, consider an adversary $A$ interacting with the modified circuit in Hybrid 2. Because the paths in $\Delta SP$ are sampled from $\mathcal{D}_{\Delta SP}$, which is statistically indistinguishable from the identity effect of $U$ and $U^{\dagger}$ pairs, the adversary’s probability of distinguishing the modified implementation from the original Hybrid 1 remains negligible. The statistical distance between these probabilities is governed by the distance between $\mathcal{D}_{\Delta SP}$ and the identity path distribution $\mathcal{I}(U_{\text{rand}})$, which is at most negligible in $\lambda$. 
Combining with Theorem~\ref{theo_protest}, it follows that for any adversary $A$ in Hybrid 2, there exists a negligible function $negl$, such that for all $\lambda$, 
${\Pr\left\lbrack {A~wins~in~Hybrid~2} \right\rbrack} \leq \frac{1}{2^{\lambda}} + \ell^{2}k/\left| {n + \lambda} \right|$

Hence,
$ \left| \Pr\left\lbrack  {A~wins~in~Hybrid~2} \right\rbrack \right. 
 - \left. \Pr\left\lbrack {A~wins~in~Hybrid~1} \right\rbrack \right| \leq \text{negl}(\lambda).$
\end{proof}

Hybrid 3: Same as Hybrid 1, except the adversary is allowed to make arbitrarily many queries.
\begin{lemma}
\label{lem_1_3}
For any adversary $A$ in Hybrids 1 and 3, there exists a negligible function $negl$ such that for all $\lambda$, \resizebox{0.5\textwidth}{!}{$\left| {{\Pr\left\lbrack {A~wins~in~Hybrid~3} \right\rbrack} - {\Pr\left\lbrack {A~wins~in~Hybrid~1} \right\rbrack}} \right| \leq \text{negl}(\lambda).$}
\end{lemma}

\begin{proof}
We adopt the `admissible oracle lemma' \cite{GJMZ23}, a tool that allows the security of multiple queries to be reduced to the security of a single query. An adversary-challenger distinguishing game, known as the $(W,~\Pi)$-distinguishing game \cite{GJMZ23}, is parameterized by binary observable $W$ and commuting projection $\Pi$.
\begin{definition}
($(W,~\Pi)$-Distinguishing Game): Let $(A, B)$ be two quantum registers. $W$ is a binary observable, and $(\Pi)$ is a projection operator acting on $(A, B)$ such that $(\Pi)$ commutes with W. Consider the following distinguishing game,
\begin{itemize}
\item The adversary sends a quantum state to the challenger, where this quantum state resides on the two quantum registers $(A, B)$.
\item The challenger chooses a random bit $\left. b\leftarrow\left\{ 0,~1 \right\} \right.$. Subsequently, it performs the measurement $\left\{ \Pi,~I~ - \Pi \right\}$. If the measurement is rejected, the game is aborted, and a random bit $\left. ~b^{\prime}\leftarrow\left\{ 0,~1 \right\} \right.$ is output. Otherwise, the challenger applies $W^b$ to $(A, B)$ and returns $B$ to the adversary.
\item The adversary outputs a guessed bit $b^{\prime}$, where the adversary's distinguishing advantage is $\left| Pr\left\lbrack b^{'} = b \right\rbrack - 1/2 \right|$.
\end{itemize}

\end{definition}

\begin{lemma}
(Admissible Oracle): Let $G$ be an admissible unitary operator, such that $G$ commutes with $W$ and $\Pi$, and $G$ acts identically on $I-\Pi$, i.e., $G(I - \Pi) = I - \Pi$. Even if the adversary is given oracle access to $G$, each adversary achieves a negligible advantage in the $(W, \Pi)$-distinguishing game.
\end{lemma}
$G$ is the unitary operator acting as a quantum circuit $C$ on the range of $\Pi$, and acting as the identity operator $I$ on the range of $1-\Pi$.
If the $(W, \Pi)$-distinguishability game is hard, meaning the distinguishing advantage is negligible, then the game remains hard when the adversary is given oracle access to any admissible unitary operator $G$. 
Hence, \resizebox{0.43\textwidth}{!}{
$\left| {{\Pr\left\lbrack {A~wins~in~Hybrid~3} \right\rbrack} - {\Pr\left\lbrack {A~wins~in~Hybrid~1} \right\rbrack}} \right| \leq \text{negl}(\lambda).$}
\end{proof}

Hybrid 4 is the same as Hybrid 2, except that the adversary can make an arbitrary number of queries.
\begin{lemma}
For any adversary $A$ in Hybrids 2 and 4, there exists a negligible function $negl$ such that for all $\lambda$, \resizebox{0.5\textwidth}{!}{$| {\Pr\left\lbrack {A~wins~in~Hybrid~4} \right\rbrack} - {\Pr\left\lbrack {A~wins~in~Hybrid~2} \right\rbrack} |  \leq \text{negl}(\lambda).$}
\end{lemma}
\begin{proof}
In the Clifford+$R_Z$ circuit, each phase gate contributes exactly one term to the phase polynomial. Assuming that each query inserts one phase polynomial, the adversary making q queries will insert $q$ polynomials, where $q = poly(\lambda)$. By combining with Theorem \ref{theo_protest}, for any adversary $A$ in Hybrid 4, there exists a negligible function $negl$, such that for all $\lambda$, ${\Pr\left\lbrack {A~wins~in~Hybrid~4} \right\rbrack} \leq \frac{1}{2^{\lambda}} + q^{2}k/\left| {n + \lambda} \right|$.

According to Lemma \ref{lem4} and \ref{lem_1_3}, 
\resizebox{0.43\textwidth}{!}{
$| {\Pr\left\lbrack {A~wins~in~Hybrid~4} \right\rbrack} { -\Pr\left\lbrack {A~wins~in~Hybrid~2} \right\rbrack} |  \leq \text{negl}(\lambda).$}
\end{proof}
Thus the proof of Theorem \ref{theo_sec} is complete.
\end{proof}
In this context, we have not considered encrypting the output results $y$, as our focus is on the functional encryption of $f$, specifically the indistinguishability of quantum implementation. It is noteworthy that by leveraging the obfuscator as a foundational primitive, one can extend its utility to construct quantum homomorphic encryption protocols \cite{Zhang21,2023Pseudo}. This extension enables further achievements, such as quantum delegated computations with encrypted quantum states.

\begin{figure*}[hbt]
  \centering
  \subfloat[Draper Adder]{\includegraphics[width=0.34\linewidth]{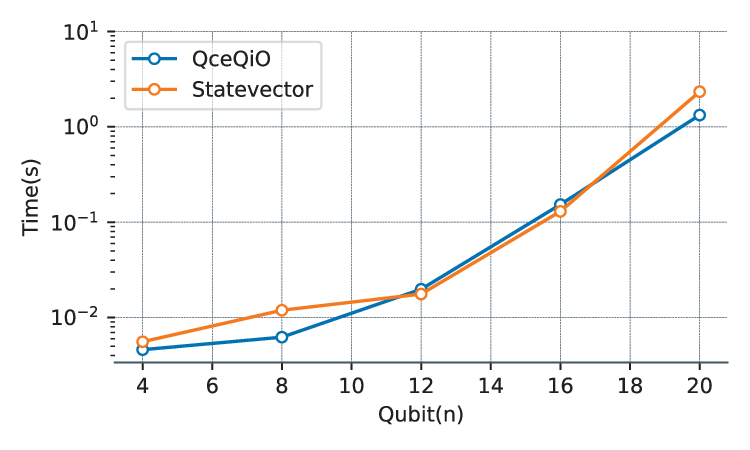}\label{fig:a}}
  \subfloat[QFT]{\includegraphics[width=0.34\linewidth]{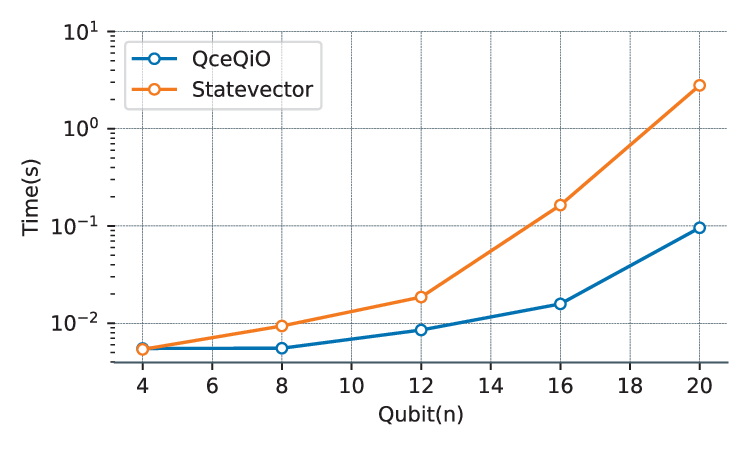}\label{fig:b}}
  \subfloat[HSP]{\includegraphics[width=0.34\linewidth]{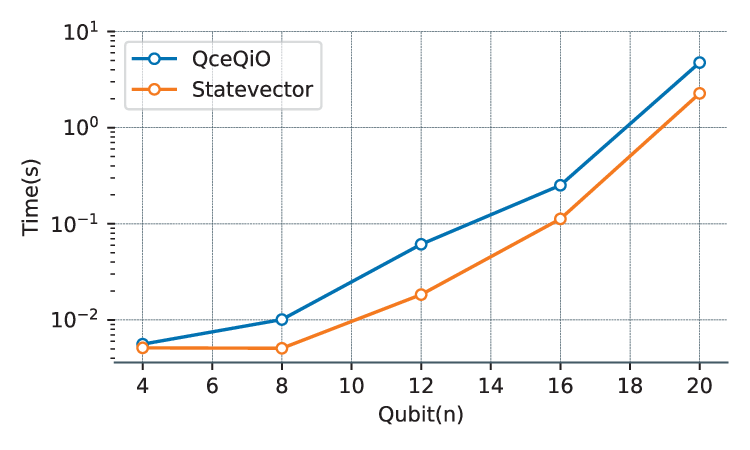}\label{fig:c}}
\newline
  \subfloat[Layered QFT]{\includegraphics[width=0.4\linewidth]{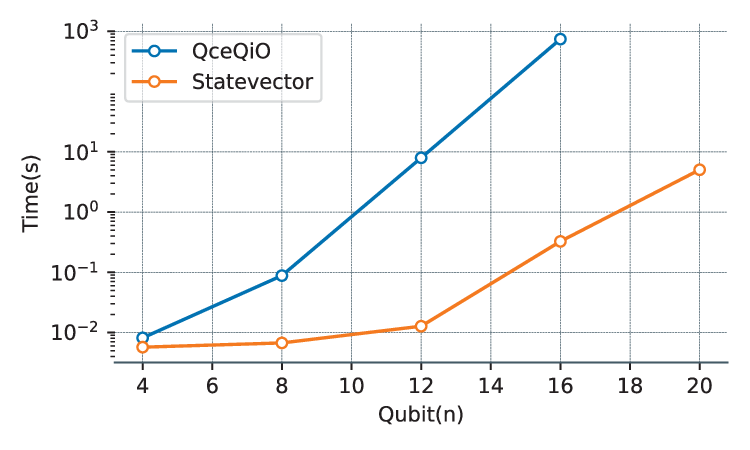}\label{fig:b}}
  \subfloat[Layered Hadamard]{\includegraphics[width=0.4\linewidth]{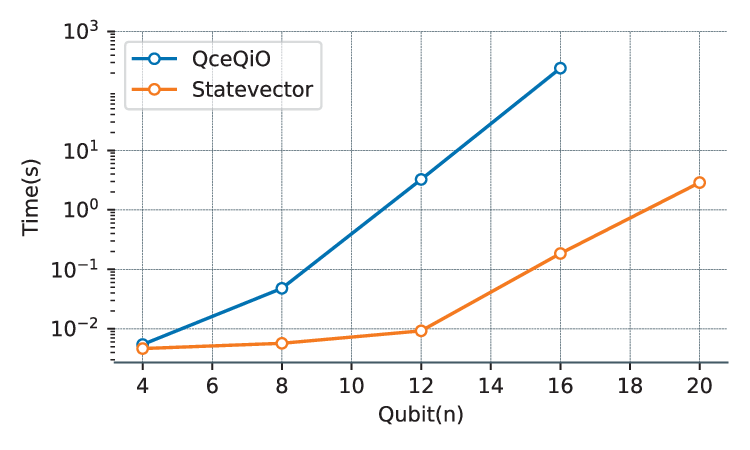}\label{fig:c}}
  \caption{Runtime graph for all test circuits. Each data point represents the average of 5 runs, and random generation is performed anew with each run.}
\label{fig12}
\end{figure*}

\section{Performance Analysis}
\label{section6}

To demonstrate the performance of QceQiO in the previous section, wel compare QceQiO with the method of state vectors under matrix semantics and perform equivalence tests.

Indistinguishability obfuscator QceQiO and the functional equivalence distinguisher run in a Linux C++ environment on Ubuntu 22.04.3 LTS, hosted on a server with an 8-core Intel Core i7 2.30GHz processor and 16GB RAM. The distinguisher primarily utilizes the equivalence checking algorithm based on Feynman path sums \cite{Amy2018}.

\textbf{The performance between quantum circuits before and after obfuscation.}
Five representative quantum circuits are selected to demonstrate performance benchmarks before and after obfuscation, Draper Adder, HSP (Hidden Subgroup Problem), QFT (Quantum Fourier Transform), Layered Hadamard, and Layered QFT. Draper Adder, HSP, and QFT are all part of the generic circuits used to build Shor's prime number factorization algorithm. Layered Hadamard and Layered QFT are included to showcase performance with more path variables, as quantum algorithms known to achieve exponential speedup over classical algorithms often begin and end with Hadamard or QFT operations.


The Layered Hadamard circuit for $n$ qubits consists of $2n$ Hadamard layers surrounding a set of randomly generated $n$ Toffoli gates. The Layered QFT circuit for $n$ qubits comprises two QFT surrounding a set of randomly generated $n$ Toffoli gates. The HSP circuit for $n$ qubits includes Hadamard transforms on register $a$, $n$ randomly generated Toffoli gates controlled by $a$ to $b$, and a QFT on register $a$. The specific circuits are shown in Fig. \ref{fig11}. 
The random generation here refers to the selection of control and target qubits being randomly chosen. The newly generated circuit is created solely for testing purposes and may not necessarily have practical arithmetic significance. Additionally, during testing, randomization is applied anew with each run.

The testing of the original circuits employs the matrix semantics under the state-vector method, which stores the entire $2^n$-sized state space of an $n$-qubit system in memory and updates it for each applied quantum gate. For each gate application, iterating over all $2^n$ elements in the memory space is required. Thus, for a circuit with $\alpha$ gates, the time and space complexity of the state vector are $O\left( \alpha 2^{n} \right)$ and $O\left(2^{n} \right)$, respectively. After QceQiO obfuscation, the quantum circuits use the construction method based on subpath sums, resulting in time and space complexities of $O\left( {(\alpha - h)2^{h}} \right)$ and $O(n+h)$, respectively. Here, $h$ represents the number of branching gates in the Clifford gates, such as the $H$ gate.

The testing results of the original and obfuscated quantum circuits are shown in Fig. \ref{fig12}, where each data point represents the average of 5 runs. It can be observed that, for the Draper Adder, HSP, and QFT circuits, as the number of qubits $n$ increases, the obfuscated circuits can almost match the runtime of the original circuits. Although, due to the additional presence of many $H$ gates, the obfuscated circuits for Layered Hadamard and Layered QFT take longer to run compared to their original counterparts, they show significant advantages in memory space usage.

\textbf{The equivalence between quantum circuits before and after obfuscation.}
Table \ref{tab:verify_result} provides verification results of the functional equivalence of quantum circuits after the QceQiO obfuscation. The $n$ column lists the number of qubits, the $m$ column lists the number of path variables, and the Clifford and $T$ columns list the respective counts of quantum gates. The positive verification and negative verification times measure the system's runtime for determining equivalence and non-equivalence, respectively. 

The benchmarks includes circuits obtained from the literature or adjusted from publicly available circuits. All the benchmark sets used are frequently employed in quantum circuit optimization and equivalence checking, as documented in the suites like \cite{2023Synthesizing,2022ZX}. Specifically, the benchmarks consist of implementations of reversible arithmetic and rigorous quantum algorithms, as these circuits are considered a primary source of complexity in quantum algorithms. Among them, the oracle for the Grover algorithm is defined as $f(x)=\neg x_1\bigwedge\neg x_2\bigwedge x_3\bigwedge x_4\bigwedge\neg x_5$, and the quantum Fourier transform for 4 qubits is approximated using higher-order rotation gates on Clifford+$T$.

All benchmark tests show successful outcomes in terms of functionality equivalence. The largest benchmark circuit, HWB$_8$, with 12 bits, 2282 path variables, and over 26,000 gates, took approximately 4.6 hours to complete, while the rest of the circuits were completed within one hour, with the majority finishing within one minute. Moreover, when randomly adding or removing selected gates from the quantum circuits obfuscated by QceQiO and retesting them with the distinguisher, all circuits were proven to be inequivalent. The time required for both positive and negative verification results is roughly the same.
\renewcommand{\arraystretch}{1.1} 
\begin{table}
\caption{Verification Results after QceQiO.  \label{tab:table1}}

\centering
\scalebox{0.65}{
\begin{tabular}{|c||c||c||c||c||c||c|}
\hline
\multirow{2}{*}{Benchmarks} & \multirow{2}{*}{n} & \multirow{2}{*}{m} & \multirow{2}{*}{Clifford} & \multirow{2}{*}{T} & \multicolumn{2}{c|}{Time(s)} \\
\cline{6-7}
& & & & & Positive & Negative \\

\hline
$\mathrm{Adder_8}$	&24	&160	&1419	&614	&4.271	&4.295 \\
\hline
$\mathrm{Toff\_Barenco_3}$&	5	&12	&66	&44	&0.003	&0.003\\
\hline
$\mathrm{Toff\_Barenco_4}$&	7&	20&	127&	84&	0.003	&0.004\\
\hline
$\mathrm{Toff\_Barenco_5}$&	9	&28	&160	&124&	0.007	&0.007\\
\hline
$\mathrm{Toff\_Barenco_{10}}$&	19	&68	&493&	324&	0.045&	0.046\\
\hline
$\mathrm{CSLA\_MUX_3}$&	15	&40	&289	&132	&0.022	&0.022\\
\hline
$\mathrm{CSUM\_MUX_9}$&	30&	56&	638&	280&	0.215&	0.191\\
\hline
$\mathrm{GF(2^4)\_Mult}$&	12&	28	&263&	180	&0.009&	0.009\\
\hline
$\mathrm{GF(2^5)\_Mult}$&	15&	36&	393&	286&	0.017&	0.016\\
\hline
$\mathrm{GF(2^6)\_Mult}$&	18	&44&	559	&402	&0.027	&0.027\\
\hline
$\mathrm{GF(2^7)\_Mult}$&	21	&52	&731&	560&	0.042&	0.043\\
\hline
$\mathrm{GF(2^8)\_Mult}$&	24	&60	&975	&712	&0.087	&0.082\\
\hline
$\mathrm{GF(2^9)\_Mult}$&	27	&68&	1179	&918	&0.096&	0.088\\
\hline
$\mathrm{GF(2^{10})\_Mult}$&	30	&76	&1475	&1110	&0.135	&0.132\\
\hline
$\mathrm{GF(2^{16})\_Mult}$&	48	&124&	3694	&2832	&0.766&	0.699\\
\hline
$\mathrm{GF(2^{32})\_Mult}$&	96&	252&	14259	&11296&	11.605&	12.137\\
\hline
$\mathrm{GF(2^{64})\_Mult}$&	192&	508&	55408	&45120&	139.436	&132.892\\
\hline
$\mathrm{GF(2^{128})\_Mult}$&	384	&1020	&231318	&180352&	2162.481&	2110.475\\
\hline
$\mathrm{Grover_5}$&	9&200	&1515&490	&0.971 &0.917\\
\hline
$\mathrm{Hamming_{15}\_high}$&	20	&716	&5332	&3462	&74.322	&75.451\\
\hline
$\mathrm{Hamming_{15}\_low}$&17&	76&	612&	158	&0.223	&0.203\\
\hline
$\mathrm{Hamming_{15}\_med}$&17&	184	&1251	&762	&1.917&	1.799\\
\hline
$\mathrm{HWB_6}$&7&	52&	369	&180&	0.191	&0.224\\
\hline
$\mathrm{HWB_8}$&12	&2282	&17583	&8895	&16655.021	&17415.308\\
\hline
$\mathrm{Mod\_Adder_{1024}}$&28&	660&	4363&	3006&	64.456	&57.861\\
\hline
$\mathrm{Mod\_Mult_{55}}$&9	&28	&180	&84	&0.009&	0.009\\
\hline
$\mathrm{Mod\_Red_{21}}$&11&	60&	392&	192	&0.04	&0.041\\
\hline
$\mathrm{Mod5_4}$&5	&12	&66&	44&	0.002	&0.002\\
\hline
$\mathrm{QCLA\_Adder_{10}}$&36	&100&	627	&400	&1.216	&1.136\\
\hline
$\mathrm{QCLA\_Com_7}$&24&	74	&1237	&297&	1.181&	1.471\\
\hline
$\mathrm{QCLA\_Mod_7}$&26	&164&	1641	&650	&48.569	&44.402\\
\hline
$\mathrm{QFT_4}$&5&	84	&218	&136&	0.048&	0.048\\
\hline
$\mathrm{RC\_Adder_6}$&14&	44&	322	&124&	0.068&	0.073\\
\hline
$\mathrm{Toff_3}$&5	&12	&52	&36	&0.002	&0.002\\
\hline
$\mathrm{Toff_4}$&7	&20&	87&	58	&0.003&	0.003\\
\hline
$\mathrm{Toff_5}$&9	&18&	112	&80	&0.005&	0.005\\
\hline
$\mathrm{Toff_{10}}$&19&	68	&297	&190	&0.02&	0.024\\
\hline
$\mathrm{VBE\_Adder_3}$&10	&20	&167&	94	&0.011&	0.011\\
\hline
\end{tabular}}
\label{tab:verify_result}
\end{table}

\section{Conclusion}
\label{section7}
In this paper, we solved the urgent problem for confidentiality and proprietary of quantum circuits in commercially sensitive delegated computing environments. We proposed a quantum indistinguishable obfuscation scheme based on the equivalence of quantum circuits, overcoming restriction observed in previous task-oriented approaches. It not only contributes to advancing the field of quantum information security but also provides insights into practical applications such as secure delegated computing and intellectual property protection in quantum environments. 

Like other quantum obfuscation and quantum computing verification schemes, the proposed scheme requires verifiers to have a quantum computer of constant size. This is because verifiers can transform classical functions into quantum implementations, requiring the ability to prepare the necessary quantum states in quantum networks.
Future research will focus on optimizing and restricting the quantum capabilities of verification parties using techniques such as quantum teleportation and semi-quantum.

\bibliographystyle{quantum}

\bibliography{bare_jrnl_new}
%
%
%
%

\end{document}